\documentclass[%
 aip,
 amsmath,amssymb,
 reprint,%
]{revtex4-1}

\usepackage{graphicx}% Include figure files
\usepackage{dcolumn}% Align table columns on decimal point
\usepackage{bm}% bold math
\usepackage[T1]{fontenc}
\usepackage{mathptmx}
\usepackage{etoolbox}

\usepackage{comment}
\usepackage{t1enc}
\usepackage[latin1]{inputenc}
\usepackage[english]{babel}
\usepackage{url}
\usepackage{theorem,xspace}
\usepackage{graphicx}
\usepackage{mathrsfs}
\usepackage{amsmath,amssymb}
\usepackage[all,poly,dvips]{xy}
\usepackage{todonotes}
\usepackage[ruled,vlined]{algorithm2e}
\usepackage{multirow}
\usepackage{float}
\usepackage{xcolor}
\usepackage{tikz}
\usetikzlibrary{arrows, chains, fit, quotes}
\usetikzlibrary{decorations.pathmorphing}
\usetikzlibrary{calc}
\usetikzlibrary{shapes.geometric}
\usetikzlibrary{angles}
\usetikzlibrary{quotes}
\tikzset{
  graph/.style={
    draw,
    circle
 }
 arc/.style={
    >=latex
 }
}

\tikzset{block/.style={rectangle, draw, max text width=6em, text badly centered, node distance=3cm, inner sep=4pt}}

\floatstyle{ruled}
\newfloat{Algorithm}{htpb}{alg}

\newcommand{\set}[1]{\{\,#1\,\}}

\newtheorem{teo}{Theorem}[section]
\newtheorem{theorem}{Theorem}[section]

\newtheorem{defi}{Definition}[section]

\newcounter{noqed}
\newcommand{\qed}{ \ifmmode\mbox{ }\fi\rule[-.05em]{.3em}{.7em}\setcounter{noqed}{0}}
\newenvironment{proof}[1][{}]{\noindent{\bf Proof#1. }\setcounter{noqed}{1}}{\ifnum\value{noqed}=1\qed\fi\par\medskip}

\newcommand{\mineq}{\approx}
\newcommand{\ami}{\operatorname{AMI}}

\newcommand{\tld}[2]{\widetilde{#1}^{#2}}

\newcommand{\view}[2]{\widetilde{#1}^{#2}}

\newcommand{\Aut}[1]{\operatorname{\textbf{Aut}}(#1)}

\renewcommand{\emptyset}{\varnothing}
\renewcommand{\phi}{\varphi}
\renewcommand{\epsilon}{\varepsilon}

\newcommand{\quot}[2]{#1/#2}

\newcommand{\n}[1][1pt]{\cir<#1>{}}

\makeatletter
\def\@email#1#2{%
 \endgroup
 \patchcmd{\titleblock@produce}
  {\frontmatter@RRAPformat}
  {\frontmatter@RRAPformat{\produce@RRAP{*#1\href{mailto:#2}{#2}}}\frontmatter@RRAPformat}
  {}{}
}%
\makeatother

\begin{document}

\preprint{AIP/123-QED}

\title[Quasifibrations of Graphs to Find Symmetries in Biological Networks]{Quasifibrations of Graphs to Find Symmetries in Biological Networks}

\author{Paolo Boldi}
    \affiliation{Computer Science Department, Universit\`a degli Studi di Milano, Milan, Italy.}
\author{Ian Leifer}
    \affiliation{Levich Institute and Physics Department, City College of New York, New York, NY 10031}
\author{Hern\'an A. Makse}
    \affiliation{Levich Institute and Physics Department, City College of New York, New York, NY 10031}
 \email{hmakse@ccny.cuny.edu.}

\date{\today}% It is always \today, today,
             %  but any date may be explicitly specified

\begin{abstract}
A \emph{fibration} of graphs is an homomorphism that is a local
isomorphism of in-neighbourhoods, much in the same way a covering
projection is a local isomorphism of neighbourhoods.  Recently, it has
been shown that graph fibrations are useful tools to uncover
symmetries and synchronization patterns in biological networks ranging
from gene, protein, and metabolic networks to the brain. However, the
inherent incompleteness and disordered nature of biological data
precludes the application of the definition of fibration \emph{as it is}; as a consequence, also the
currently known algorithms to identify
fibrations fail in these domains. In this paper, we introduce and develop systematically the theory of
quasifibrations which attempts to capture more realistic patterns of
almost-synchronization of units in biological networks. We provide an
algorithmic solution to the problem of finding quasifibrations in
networks where the existence of missing links and variability across
samples preclude the identification of perfect symmetries in the
connectivity structure.  We test the algorithm against other
strategies to repair missing links in incomplete networks using real
connectome data and synthetic networks. Quasifibrations can be
applied to reconstruct any incomplete network structure characterized
by underlying symmetries and almost synchronized clusters.
\end{abstract}

\maketitle

\begin{quotation}
Symmetry is a foundational principle underlying the behavior of numerous natural systems. Biological networks have been recently shown ~\cite{morone2020PNAS,leifer2020PLOS,morone2019NatComm,LeiferBMC20} to exhibit the symmetry stemming from the existence of graph fibrations. We refer to this symmetry as fibration symmetry. However, experimental observations providing the basis for constructing
biological networks are inherently imperfect, as biological systems
are noisy and disordered, and perfect symmetries are never realized in
biology. To address this problem, we introduce the concept of
quasifibrations to describe more realistic patterns of 
synchronization in biological networks. We provide an optimization algorithm
to obtain quasifibrations in graphs along with a network
reconstruction procedure of an incomplete network with missing links
to restore the perfect ideal symmetry in the network.
\end{quotation}

\section{Introduction}
\label{sec:intro}

A large part of modern physics is built on the principles of symmetry,
typically expressed using the mathematical theory of groups and group
actions spanning from elementary particles and fundamental forces to
the state of matter. Symmetries and group theory are at the core of
the theoretical formulation of the Standard Model, which provides a
unified description of all existing particles and their interactions
(exclusion made for gravity).

While the concept of symmetry is at the basis of such important
conceptual advances in physics, it is not (yet) as pervasively adopted
in biology with the exceptions of the work of D'Arcy
Thompson~\cite{darcy} and the discussions emerging from the Eleventh
Nobel Symposium held in 1968, chaired by Jacques Monod ~\cite{monod}.
In our opinion, two crucial aspects probably contributed to determine
the reduced attention to the role played by symmetries in biological
systems:
\begin{itemize}
  \item conventionally, the symmetries of a network are identified with the
permutations of the nodes that leave the adjacency matrix
invariant (the so-called \emph{network automorphisms}); this form of symmetry is very rigid and
it is rarely observed in biology; moreover, automorphisms are not the right tool to describe the notion of synchronization,
that is the type of symmetry that more often occurs in living networks;
  \item even playing with more general kinds of symmetry, biological systems are never exactly symmetrical: the high
redundance of biological systems leaves space to small asymmetries, that are influential to the overall working of the system
but fail to be captured by a rigorous mathematical definition. 
\end{itemize}  

%monteiroAlgorithm It can be identified using a slightly modified version of the Paige and Tarjan algorithm \cite{paige1987} introduced in \cite{monteiroAlgorithm}.
Recent works~\cite{morone2019NatComm, morone2020PNAS, leifer2020PLOS, LeiferBMC20}
provided credible evidence that the right notion of symmetry exhibited
by biological networks is that induced by \emph{graph fibrations}. Existence of graph fibrations have been shown to produce the synchronized solutions in a network \cite{stewart2006, lerman2013, nijholt2016} that were observed on the experimental data \cite{LeiferBMC20}. In this paper we shall relax the definition of fibration symmetry to allow for errors, introducing the notion of \emph{quasifibration}. Quasifibrations are a more realistic way to describe synchronization patterns in living networks. We shall also provide practical techniques to reconstruct symmetries in spite of errors, and discuss whether and how such techniques can be used to study real-world biological systems. Reference \cite{Leifer21a} presents a more general approach to the graph reconstruction problem that can be applied to both directed and undirected graphs, however the approach presented here is considerably faster and is easier to scale.

The paper is organized as follows. Section \ref{sec:defs} sets up the
basic definitions and reviews the previous literature on graph
fibrations.  Section \ref{sec:auto} further elaborates on previous
results on symmetry identification, comparing graph fibrations with
the stronger notion of automorphism.  Section \ref{sec:ufc} reviews
the concepts of input trees, universal and minimal
graph fibrations which are necessary to generalize the problem to
quasifibrations.  Section \ref{sec:tools} then introduces the concept
of quasifibration and the proposed algorithmic steps to identify
quasifibrations in networks and repair the missing or excess links in 
networks with symmetries.  Finally, section
\ref{sec:experiments} shows computational experiments testing the
quasifibration algorithm against different strategies to repair
networks using synthetic data and real data on previously analyzed
connectomes. We conclude the paper with summary and outlook.
All the algorithms are available at \url{https://github.com/boldip/qf}.

\section{Definitions and basic properties}
\label{sec:defs}

\subsection{Graph-theoretical definitions}

\paragraph{Graphs.}
A \emph{(directed multi)graph} (or ``network'') $G$ is defined by a set $N_G$ of nodes and a set
$A_G$ of arcs, and by two functions $s_G,t_G:A_G\to N_G$ that specify the
source and the target of each arc (we shall drop the subscripts whenever no
confusion is possible). Note that, differently from other definitions commonly used
in the literature, arcs have a direction (that is why our graphs are ``directed'') and
we may have multiple parallel arcs with the same source and the same target (that is why 
``multigraphs''). Unless otherwise specified, both $N$ and $A$ are finite.
An example of a graph is shown in Figure~\ref{fig:graph-example}.

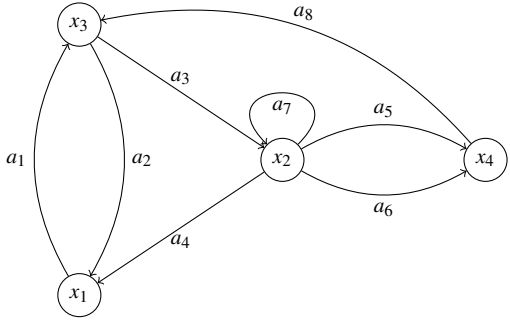
\begin{figure}
    \centering
    \begin{tikzpicture}[scale=0.9, every node/.style={scale=0.9}]
	\node[graph, circle, draw] (x1) at (0,0) {$x_1$};
	\node[graph, circle, draw] (x2) at (3,2) {$x_2$};
	\node[graph, circle, draw] (x3) at (0,4) {$x_3$};
	\node[graph, circle, draw] (x4) at (6,2) {$x_4$};
	
	\draw[->,bend left] (x1) to node[left] {$a_1$} (x3);
	\draw[->,bend left] (x3) to node[right] {$a_2$} (x1);
	\draw[->] (x3) to node[above] {$a_3$} (x2);
	\draw[->] (x2) to node[below] {$a_4$} (x1);
	\draw[->,bend left] (x2) to node[above] {$a_5$} (x4);
	\draw[->,bend right] (x2) to node[below] {$a_6$} (x4);
	\draw[loop, ->] (x2) to[below] node {$a_7$} (x2);
	\draw[->,bend right] (x4) to node[above] {$a_8$} (x3);
\end{tikzpicture}
    \caption{\label{fig:graph-example}
      An example of graph, with four nodes and eight
      arcs: nodes are represented by circles, arcs by arrows stemming from
      their source and ending at their target. Observe that there are two
      parallel arcs in $G(x_2, x_4)=\set{a_5, a_6}$ and there is a loop
      $a_7$ at $x_2$.}
\end{figure}

We use the notation $G(x,y)$ for denoting the set of
arcs from $x$ to $y$, that is, the set of arcs $a\in A_G$ such that $s(a)=x$ and
$t(a)=y$. A \emph{loop} is an arc with the same source and target. 
Following common usage, we denote with $G(-,x)$ the set of arcs coming
into $x$, that is, the set of arcs $a\in A_G$ such that $t(a)=x$, and analogously with
$G(x,-)$ the set of arcs going out of $x$.

% A \emph{symmetric graph} is a graph endowed with a symmetry, that is, an
% involution (a self-inverse bijection) $(\ovr{\phantom{a}}):A_G\to A_G$ such that
% $s(a)=t(\ovr a)$ (and consequently $t(a)=s(\ovr a)$) for all arcs $a\in A_G$. A
% \emph{semi-edge} of a symmetric graph is a loop $a$ such that $\ovr a=a$.  Given a
% graph $G$, we define its \emph{(formal) symmetrization} $\sym G$ as the graph
% obtained by adding for each arc $a \in G(x,y)$ a new arc $\ovr a$ going from $y$ to
% $x$, with symmetry defined in the obvious way.
% 
% A graph $G$ is \emph{$j$-inregular} (\emph{$k$-outregular}) if $|G(-,x)|=j$
% ($|G(x,-)|=k$, respectively). A $j$-inregular, $k$-outregular graph is said to
% be \emph{$(j,k)$-regular}. For finite or symmetric graphs $(j,k)$-regularity
% implies $j=k$, and when $j=k$ we simply say that $G$ is $j$-regular. 

A \emph{path} (of length $n$) is a sequence $x_0,a_1,x_1,\cdots,x_{n-1},a_n,
x_n$ alternating nodes $x_i\in N_G$ and arcs $a_j\in A_G$ in such a way that $s(a_j)=x_{j-1}$ and $t(a_j)=x_j$. We
shall usually omit the nodes from the sequence when at least one arc is present.
% If $G$ is symmetric, a path is called \emph{symmetrically stuttering} (or,
% simply, stuttering) iff it contains a subpath of the form $a\ovr a$; a
% \emph{nonstuttering walk} of a graph $G$ is a nonstuttering path of $\sym G$.
% Since we shall only be concerned with walks of this kind, we shall drop the
% adjective ``nonstuttering'' in the sequel. 
We shall say that $G$ is
\emph{(strongly) connected} iff for every choice of $x$ and $y$ there is a
path from $x$ to $y$; the \emph{diameter} $D_G$ of a strongly connected
graph is the maximum length of a shortest path between two nodes.

The graph of Figure~\ref{fig:graph-example} is strongly connected; an example of 
a path from $x_1$ to $x_2$ is represented by the sequence $x_1,a_1,x_3,a_3,x_2,a_7,x_2,a_5,x_4,a_8,x_3,a_3,x_2$. There are, in fact,
infinitely many paths from $x_1$ to $x_2$.

We shall occasionally deal with node- (or arc-) coloured graphs: a \emph{node-coloured graph} (arc-coloured graph, respectively)
with set of colours $C$ is a graph endowed with a colouring function
$\gamma: N_G\to C$ ($\gamma: A_G \to C$, respectively). 
% For symmetric graphs, we require that there is an involution
% $(\ovr{\phantom{a}}):C\to C$ such that $\gamma(\ovr a)=\ovr{\gamma(a)}$.  A (coloured)
% graph is \emph{separated} iff it has no parallel arcs (with the same colour).
% The name originates from the fact that such graphs are separated for the double
% negation topology in the topos of (coloured) graphs---see~\cite{VigGTTG}.

\paragraph{Graph homomorphisms.}
For standard directed graphs, homomorphisms are defined only on nodes, with the requirements
that existing arcs are preserved (i.e., if there was an arc between two nodes in the starting graph,
then there must be an arc between the images of these two nodes in the ending graph).
Extending this definition to the case of multigraphs requires some care, because we can
decide separately how to map nodes and arcs, but we must ensure that the two maps agree.
Formally, a \emph{graph homomorphism} $\xi:G\to H$ is given by a pair of
functions $\xi_N:N_G\to N_H$ and $\xi_A:A_G\to A_H$ commuting with the
source and target maps, that is, for all $a \in A_G$,
$s_H(\xi_A(a))=\xi_N(s_G(a))$ and $t_H(\xi_A(a))=\xi_N(t_G(a))$
(again, we shall drop the subscripts whenever no confusion is
possible). As we said, the meaning of the commutation conditions is that 
a homomorphism maps nodes to nodes and arcs to arcs in
such a way to preserve the incidence relation. (In the case of
node or arc coloured graphs, we also require $\xi_N$ or $\xi_A$ to commute with the colouring
function.)
%A morphism between
%symmetric graphs is \emph{symmetric} iff it commutes with the symmetries.  
A homomorphism is \emph{surjective}\footnote{Often called an \emph{epimorphism} in the literature.} iff
$\xi_N$ and $\xi_A$ are both surjective. An example of graph homomorphism
is shown in Figure~\ref{fig:morph-example}.

\begin{figure*}
    \centering
    \begin{tabular}{ccc}
        \begin{tikzpicture}[scale=1, every node/.style={scale=0.9}]
        	\node[graph, draw, circle, fill=green] (x1) at (0,0) {$x_1$};
        	\node[graph, draw, circle, fill=red] (x2) at (3,2) {$x_2$};
        	\node[graph, draw, circle, fill=red] (x3) at (0,4) {$x_3$};
        	\node[graph, draw, circle, fill=green] (x4) at (6,2) {$x_4$};
        	
        	\draw[->,bend left, blue] (x1) to node[left] {$a_1$} (x3);
        	\draw[->,bend left, red] (x3) to node[right] {$a_2$} (x1);
        	\draw[->, cyan] (x3) to node[above] {$a_3$} (x2);
        	\draw[->, red] (x2) to node[below] {$a_4$} (x1);
        	\draw[->,bend left, red] (x2) to node[above] {$a_5$} (x4);
        	\draw[->,bend right, red] (x2) to node[below] {$a_6$} (x4);
        	\draw[loop, ->, cyan] (x2) to[below] node {$a_7$} (x2);
        	\draw[->,bend right, brown] (x4) to node[above] {$a_8$} (x3);
        \end{tikzpicture}
            & $\qquad$ &
        \begin{tikzpicture}[scale=1, every node/.style={scale=0.9}]
        	\node[graph, draw, circle, fill=green] (y1) at (0,0) {$y_1$};
        	\node[graph, draw, circle, fill=red] (y2) at (0,4) {$y_2$};
        	\node[graph, draw, circle] (y3) at (2,2) {$y_3$};
        	
        	\draw[->,bend left, blue] (y1) to node[left] {$b_3$} (y2);
        	\draw[->,bend left=80, brown] (y1) to node[left] {$b_4$} (y2);
        	\draw[->,bend left, red] (y2) to node[left] {$b_1$} (y1);
        	\draw[->,loop, cyan] (y2) to node[below] {$b_2$} (y2);
        	\draw[->] (y1) to node[below] {$b_5$} (y3);
        	\draw[->] (y3) to node[above] {$b_6$} (y2);
        \end{tikzpicture} 
    \end{tabular}
    \caption{\label{fig:morph-example}An example of a graph homomorphism
      $\xi: G \to H$. Here, we use node (arc) colours to indicate how
      nodes (arcs) are mapped: every node (arc) of $G$ is mapped to the
      only node (arc) of $H$ with the same colour; $\xi$ is not a
      surjective (the black node and arcs of $B$ do not have any
      counterimage). To see that $\xi$ is a morphism, note that the colors of source
      and target are respected; for instance all red arcs have a red source and a green target
      (because they are all mapped to $b_1$).
      Incidentally, if we think of $G$ and $B$ as being
      really coloured as in the picture, $\xi$ is a correct homomorphism
      between coloured graphs (because it commutes with the colouring
      function).}
\end{figure*}
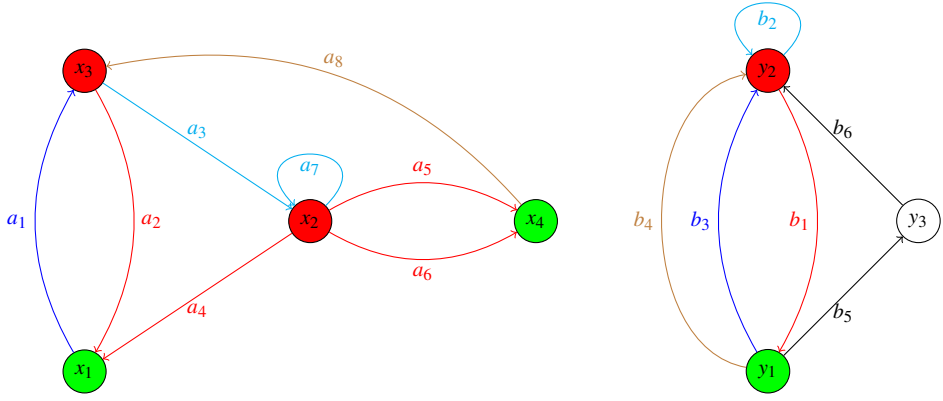

Epimorphisms that are also injective (both on nodes and on arcs) are called
\emph{isomorphisms}; we write $G \cong H$ iff there is an isomorphism between $G$ and $H$: 
isomorphisms are just a way to change the ``identity'' of nodes and arcs without
chaning the graph structure. Since we are only interested in the structural properties of a network, we normally consider
graphs only up to isomorphism (that is, we identify a graph $G$ with any other graph $H$ such that $G \cong H$).  
An isomorphism $\alpha: G \to G$ is called an \emph{automorphism} of $G$: automorphisms form a group (denoted by $\Aut G$) with respect to function
composition.  A graph with no non-trivial automorphisms (i.e., no automorphism other than the identity) is called \emph{rigid}.

\subsection{Fibrations}

The central concept we are going to deal with is that of \emph{graph
  fibration}, a particular kind of graph homomorphism. Fibrations can
be thought of as a relaxed version of isomorphisms, as we shall see in
a second.  
The categorical definition of fibration is credited to Alexandre Grothendieck~\cite{GroTDTEGAI}; it
builds upon previous similar concepts in algebraic graph theory~\cite{SacSUGG} and is closely related
to the generalization of covering in topology~\cite{PSVEGB}. The application of this definition to graphs (seen as free categories) was given
in~\cite{BoVGF} and was largely used in the context of distributed systems.   

 \begin{defi}
 \label{def:fib}
A \emph{fibration} between graphs $G$ and $B$ is a homomorphism
$\phi:G\to B$ such that for each arc $a\in A_B$ and for each node
$x\in N_G$ satisfying $\phi(x)=t(a)$ there is a unique arc $\tld ax\in
A_G$ (called the \emph{lifting of $a$ at $x$}) such that $\phi(\tld
ax)=a$ and $t(\tld ax)=x$.
 \end{defi}

We inherit some topological terminology. If $\phi:G\to B$ is a
fibration, $G$ is called the \emph{total graph} and $B$ the
\emph{base} of $\phi$. We shall also say that $G$ is \emph{fibred
  (over $B$)}. The \emph{fibre} over a node $x\in N_B$ is the set of
nodes of $G$ that are mapped to $x$, and shall be denoted by
$\phi^{-1}(x)$. A fibre is \emph{trivial} if it is a singleton, that
is, if $|\phi^{-1}(x)|=1$. A fibration is \emph{nontrivial} if at
least one fibre is nontrivial, \emph{trivial} otherwise; it is
\emph{proper} if all fibres are nontrivial.

The homomorphism of Figure~\ref{fig:morph-example} is \emph{not} a
fibration: for instance, the fibre over $y_2$ is $\set{x_2, x_3}$, but
while arcs $b_3$ and $b_4$ can both uniquely be lifted at $x_3$ ($\tld
{b_3}{x_3}=a_1$ and $\tld {b_4}{x_3}=a_8$), arc $b_2$ cannot be lifted
at all at $x_3$, while it can be lifted twice at $x_2$.

\medskip
There is an intuitive characterization of fibrations based on the
concept of local in-isomorphism. An equivalence relation $\simeq$
between the nodes of a graph $G$ is said to satisfy the \emph{local
  in-isomorphism property} if the following holds:
\begin{itemize}
\item[] {\bf Local In-Isomorphism Property:} If $x \simeq y$ there exists a
(colour-preserving, if $G$ is coloured) bijection $\psi: G(-,x) \to G(-,y)$
such that $s(a) \simeq s(\psi(a))$, for all $a\in G(-,x)$.
\end{itemize}
The following proposition shows that fibrations and surjective homomorphisms whose fibres
satisfy the previous property are naturally equivalent:
\begin{teo}[\cite{BoVGF}, Theorem 2.1]
\label{teo:lip}
Let $G$ be a graph. Then:
\begin{enumerate}
\item\label{enu:fiblip} if $\phi: G \to B$ is a fibration, then the equivalence relation on the
nodes of $G$ whose equivalence classes are the nonempty fibres of $\phi$ satisfies the
local in-isomorphism property;
\item\label{enu:lipfib} if $\simeq$ is a relation between nodes of $G$ satisfying the local in-isomorphism property, then
there exists a graph $B$ and a surjective fibration $\phi:G \to B$ whose fibres are the
equivalence classes of $\simeq$.
\end{enumerate}
\end{teo}

Another possible, more geometric way of interpreting the definition of
fibration is that given a node $x$ of $B$ and path $\pi$ terminating at $x$,
for each node $y$ of $G$ in the fibre of $x$ there is a unique path terminating
at $y$ that is mapped to $\pi$ by the fibration; this path is called the
\emph{lifting of $\pi$ at $y$}, and it is denoted by $\tld \pi y$. 

\paragraph{Intuition: System simulation via fibrations.} 
In order to provide a more intuitive grasp on the notion of fibration we 
want to briefly discuss how it is used in distributed computing.
Consider a graph $G$, whose nodes represent entities of
some kind: we will call them \emph{processors}, but our description is more general, 
and applies \emph{mutatis mutandis} to any system composed by autonomous interacting
units, or agents. Each processor has an internal state belonging to a set of possible states $X$. 
During the evolution of the system, each processor
changes its state depending on its own current state and on the states of its
in-neighbours: equivalently, arcs represent unidirectional links along which a
processor transmits messages. 

In a distributed system like this, the behaviour of each processor $x$ is partly influenced by the behaviour of its
in-neighbours and, more in general, by all the processors that can (directly or indirectly) communicate with $x$. 
If we possessed a complete knowledge of the system internals (what is the set of states $X$, what exact protocol
each processor is following, what kind of signals are transmitted, and when) we may provide a precise description of
the system \emph{global} behaviour, for example, to simulate it or to make predictions or understand correlations.
In particular, we are interested in capturing the symmetries that the system exhibits: which groups of processors 
behave alike? 

\begin{figure}
\centering
\begin{tikzpicture}[scale=0.9, every node/.style={scale=0.9}]
	\node[regular polygon, regular polygon sides=6, minimum size=5cm] at (0,0) (A) {};
	\foreach \i in {1,...,6}
		\node[circle, draw] at (A.corner \i) (X \i) {$x_{\i}$};
	\node[circle, draw] at (-1,0.6) (L) {$y_1$};
	\node[circle, draw] at (-2.2,1.4) (W) {$w$};
	\node[circle, draw] at (1,-0.6) (R) {$y_2$};
	\node[circle, draw] at (0,0) (C) {$z$};
	\foreach \i in {1,...,6} {
		\draw[->] (X \i) to (L);
		\draw[->] (X \i) to (R);
		\draw[->] (C) to (X \i);
	}
	\draw[->] (L) to (C);	
	\draw[->] (R) to (C);	
	\draw[->] (L) to (W);
\end{tikzpicture}
\caption{\label{fig:intuition-example}An example of a distributed system: During the evolution of the system, each processor
changes its state depending on its own current state and on the messages received from its
in-neighbours along the arcs, representing unidirectional communication links.}
\end{figure}
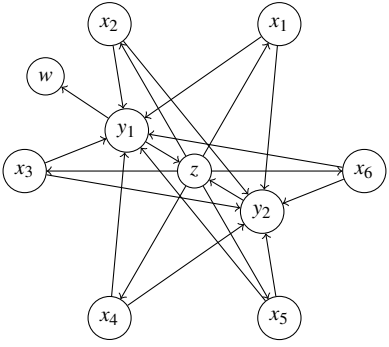

As we said, with a complete knowledge of the individual processors and their behaviour the task of determining symmetries 
is, at least in principle, possible. But what if the only knowledge we possess is the connection graph $G$?
What are the symmetries that the graph exhibits, in absence of other information? 

These questions are translated into one of the main problems of distributed systems~\cite{AngGLPNP}: establish
which configurations of states can be reached when all processors
start from the same state, change state at the same time and run the
same algorithm---or, as usually stated, when
the network is \emph{anonymous} (or \emph{uniform}) and
synchronous. The main point to be noted here is that, under such
constraints, \emph{the existence of a fibration $\phi: G\to B$ forces
  all processors in the same fibre to remain always in the same
  state}.

More precisely, every behaviour of the processors of the ``large''
graph $G$ can be simulated on the ``smaller'' graph\footnote{If we
  assume that $\phi$ is surjective but not injective on the
  nodes then $G$ really has more nodes than $B$, which explains why we
  call $G$ ``large'' and $B$ ``small''.}  $B$: the fibration describes
how nodes and arcs of $G$ are mapped to nodes and arcs of $B$.  Nodes
in the same fibre (that is, mapped to the same node of $B$) cannot
distinguish from each other because of the local in-isomorphism
property---a general homomorphism would not make the simulation
possible, because it may map together nodes with different views of the
system.  In this sense, fibrations are related to how the system can
be mapped to a smaller, but simulation-equivalent, core.

Let us look at Figure~\ref{fig:intuition-example-colored}, where we represent how the network $G$ can be fibred over a smaller network.
Each of the cyan nodes $x_i$ in $G$ is equivalent to the unique cyan node $x$ in $B$: they receive only one input (each) from the red-coloured
node (called $z$ in both networks).
Each of the two pink nodes $y_i$ in $G$ is equivalent to the unique pink node $y$ in $B$: they each receive six inputs from the cyan-coloured nodes.
In $G$ there are six such cyan-colored nodes, but they all behave alike (because they behave like the unique cyan node of $B$), so
the fact that each $y_i$ in $G$ receives inputs from six different cyan nodes, while $y$ in $B$ receives six times the input from the same node is
irrelevant (because the six cyan nodes in $G$ behave all like the single cyan node of $B$).
Similarly, $z$ in $G$ receives inputs from $y_1$ and $y_2$, whereas $z$ receives twice input from $y$.
Finally, $w$ in $G$ receives inputs from $y_1$ only, and so its behaviour can be simulated by the unique green node $w$ in $B$.

Otherwise said: $\phi$ is a local in-isomorphism. What locally each node $x$ of $G$ sees of its in-neighborhood in $G$ is the same as
what the node $\phi(x)$ sees of its in-neighborhood in $B$.
The fibres of the fibration (i.e., the colours of the nodes in $G$) represent the kind of symmetry we are looking for. 

\begin{figure*}
    \centering
        \begin{tabular}{ccc}
            \begin{tikzpicture}[scale=0.9, every node/.style={scale=0.9}]
            	\node[regular polygon, regular polygon sides=6, minimum size=5cm] at (0,0) (A) {};
            	\foreach \i in {1,...,6}
            		\node[circle, draw, fill=cyan] at (A.corner \i) (X \i) {$x_{\i}$};
            	\node[circle, draw, fill=pink] at (-1,0.6) (L) {$y_1$};
            	\node[circle, draw, fill=green] at (-2.2,1.4) (W) {$w$};
            	\node[circle, draw, fill=pink] at (1,-0.6) (R) {$y_2$};
            	\node[circle, draw, fill=red]  at (0,0) (C) {$z$};
            	\foreach \i in {1,...,6} {
            		\draw[->] (X \i) to (L);
            		\draw[->] (X \i) to (R);
            		\draw[->] (C) to (X \i);
            	}
            	\draw[->] (L) to (C);	
            	\draw[->] (R) to (C);
            	\draw[->] (L) to (W);	
            \end{tikzpicture} 
            & $\qquad$ &
            \begin{tikzpicture}[scale=0.9, every node/.style={scale=0.9}]
            	\node[circle, draw, fill=red]  at (0,0) (Z) {$z$};
            	\node[circle, draw, fill=cyan] at (0,3) (X) {$x$};
            	\node[circle, draw, fill=pink] at (0,6) (Y) {$y$};
            	\node[circle, draw, fill=green] at (2,6) (W) {$w$}; 
            	\draw[->] (X) to[bend left=15] (Y);	
            	\draw[->] (X) to[bend left=-15] (Y);	
            	\draw[->] (X) to[bend left=30] (Y);	
            	\draw[->] (X) to[bend left=-30] (Y);	
            	\draw[->] (X) to[bend left=45] (Y);	
            	\draw[->] (X) to[bend left=-45] (Y);	
            	\draw[->] (Y) to[bend left=60] (Z);	
            	\draw[->] (Y) to[bend right=60] (Z);
            	\draw[->] (Z) to (X);
            	\draw[->] (Y) to (W);
            \end{tikzpicture} \\[1em]
            $G$ & & $B$
        \end{tabular}
    \caption{\label{fig:intuition-example-colored}A graph fibration $\phi$
      between the graph $G$ of Figure~\ref{fig:intuition-example} and
      another graph $B$. The node-component of $\phi$ is represented by
      the colours on the nodes; the arc-component is irrelevant (any map
      that makes it a graph homomorphism will do).}
\end{figure*}
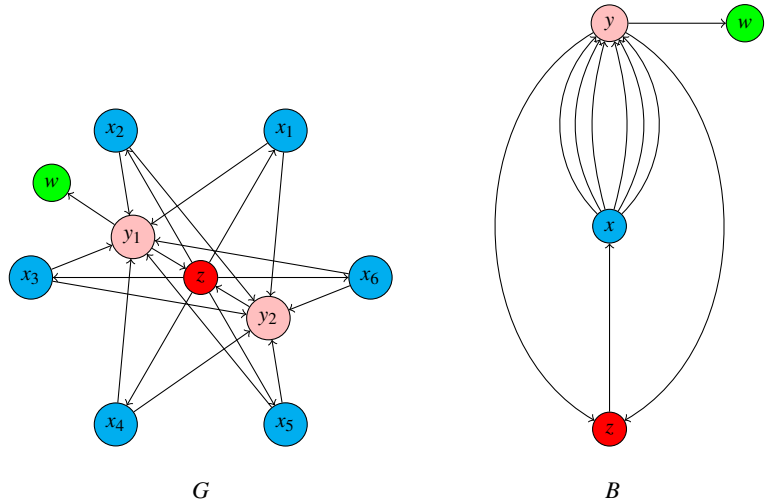

\section{Groups, fibrations and automorphisms}
\label{sec:auto}
Fibrations are a way to identify symmetries in a graph, much like automorphisms.
The reader familiar with automorphism may wonder in which sense fibrations are different from (or more useful than)
automorphisms. We are going to discuss this issue here.

A standard way to describe group symmetries is by using \emph{actions}:
a \emph{(left and faithful) action} of a group $\Gamma$ over $G$ is an injective function
$\alpha: \Gamma \to \Aut G$ that is a group homomorphism.\footnote{A function $\alpha$ between groups is a
group homomorphism iff $\alpha(xy)=\alpha(x)\alpha(y)$.}
The action induces an equivalence relation both on the nodes and
on the arcs of $G$, whose classes (the \emph{orbits} of $\alpha$) are denoted by
$\Gamma(x)$ or $\Gamma(a)$: two nodes $x,y$ are in the same orbit iff\footnote{Note that, for all $g \in \Gamma$, 
$\alpha(g)$ is an automorphism of $G$, that one can apply to both nodes and arcs.} $\alpha(g)(x)=y$ for some $g \in \Gamma$, and
similarly for arcs. The so-called \emph{quotient graph} $\quot G \alpha$ whose nodes and arcs are the orbits of nodes and
arcs under the action of $\alpha$. 

If you look at Figure~\ref{fig:intuition-example-colored}, there is a natural action $\alpha: S_6 \to \Aut G$ where $S_6$ is the
symmetric group with six generators (i.e., the group of all permutations of the set $\set{1,2,3,4,5,6}$):
every permutation $\pi \in S_6$ is associated with an automorphism that maps node $x_i$ to node $x_{\pi(i)}$,
and arc $x_i\to y_j$ to $x_{\pi(i)} \to y_j$ (for all $j=1,2$); all the other nodes and arcs are left unchanged by $\alpha$.
%There is another natural action $\beta: S_2 \to \Aut G$ that permutes $y_1$ and $y_2$.
%The two actions together provide evidence of the fact that $\Aut G \cong S_2 \times S_6$. 

This observation is an example of a general and important relation between fibrations and automorphisms, that can be described
in terms of group actions.
The action of a specific element $g \in \Gamma$ gives a bijection between 
$G(-,x)$ and $G(-,\alpha(g)(x))$ that fulfills the requirements
of the local in-isomorphism property; thus, by Theorem~\ref{teo:lip},
$\alpha$ induces a fibration $\phi: G \to B_\alpha$, where $B_\alpha$ is a
graph having as nodes the node-orbits of $\alpha$ and as many arcs
from $\Gamma(x)$ to $\Gamma(y)$ as the arcs coming into an(y) element
of $\Gamma(y)$ from elements of $\Gamma(x)$. We say that $\phi$ is
\emph{associated with $\alpha$}.  
Note that $\phi$ is in general not unique, as it depends, for every $x$ and
$y$, on the element of $\Gamma$ that is chosen to induce the local
in-isomorphism between $G(-,x)$ and $G(-,y)$.

\medskip
But not all fibrations are associated with group actions, and this is the crucial observation here:
looking again at Figure~\ref{fig:intuition-example-colored}, no action can move node $w$ (because automorphisms
preserve indegree and outdegree, and $w$ is the only node with outdegree zero), and as a consequence no action can exchange $y_1$ and
$y_2$. The fact that a \emph{fibration} identifies the two nodes $y_1$ and $y_2$ shows that fibrations give a stronger (and more
robust) definition of symmetry.

As an even more evident example, consider the 
graph $G$ on the left of Figure~\ref{fig:rigid-example}: $G$ is rigid (i.e., it has a trivial automorphism group), hence no group 
acts nontrivially on it. The reason behind this, intuitively, is the presence of the arc $a_7$ that disrupts the graph symmetry.
Nonetheless, the graph is fibred nontrivially as shown in the picture.

\begin{figure}
\centering
\begin{tabular}{ccc}
\begin{tikzpicture}[scale=0.7, every node/.style={scale=0.9}]
        % vertices
        \node[circle,draw,fill=purple] (v6) at  (1,0) {$x_6$};
        \node[circle,draw,fill=yellow] (v4) at  (0,1.5) {$x_4$};
        \node[circle,draw,fill=yellow] (v5) at  (2,1.5) {$x_5$};
        \node[circle,draw,fill=cyan] (v2) at  (0,3.5) {$x_2$};
        \node[circle,draw,fill=cyan] (v3) at  (2,3.5) {$x_3$};
        \node[circle,draw,fill=green] (v7) at  (4,3.5) {$x_7$};
        \node[circle,draw,fill=blue] (v1) at  (1,5) {$x_1$};
        %->s
        \draw[->] (v6) to node [left] {$a_1$} (v4);
        \draw[->] (v6) to node [right] {$a_2$} (v5);
        \draw[->] (v4) to node [left] {$a_3$} (v2);
        \draw[->] (v5) to node [right] {$a_4$}(v3);
        \draw[->] (v2) to[bend left] node [below] {$a_5$} (v3);
        \draw[->] (v3) to[bend left] node [below] {$a_6$}(v2);
        \draw[->] (v3) to node [above] {$a_7$} (v7);
        \draw[->] (v2) to node [left] {$a_8$} (v1);
        \draw[->] (v3) to node [right] {$a_9$} (v1);
\end{tikzpicture}
& $\qquad$ &
\begin{tikzpicture}[scale=0.9, every node/.style={scale=0.9}]
        % vertices
        \node[circle,draw,fill=purple] (y5) at  (1,0) {$y_5$};
        \node[circle,draw,fill=yellow] (y4) at  (1,1.5) {$y_4$};
        \node[circle,draw,fill=cyan] (y3) at  (1,3.5) {$y_3$};
        \node[circle,draw,fill=green] (y2) at  (4,3.5) {$y_2$};
        \node[circle,draw,fill=blue] (y1) at  (1,5) {$y_1$};
        %->s
        \draw[->] (y5) to node [left] {$b_1$} (y4);
        \draw[->] (y4) to node [left] {$b_2$} (y3);
        \draw[->,loop left] (y3) to node {$b_3$} (y3);
        \draw[->] (y3) to[bend right] node [right] {$b_4$} (y1);
        \draw[->] (y3) to[bend left] node [left] {$b_5$} (y1);
        \draw[->] (y3) to node [above] {$b_7$} (y2);
\end{tikzpicture}\\
$G$ && $B$
\end{tabular}
\caption{\label{fig:rigid-example}A rigid graph $G$, and a non-trivial fibration $\phi: G \to B$ (as usual, node colors are used to
suggest how node should be mapped, whereas the arc map is irrelevant as long as it respects the definition of graph homomorphism).}
\end{figure}
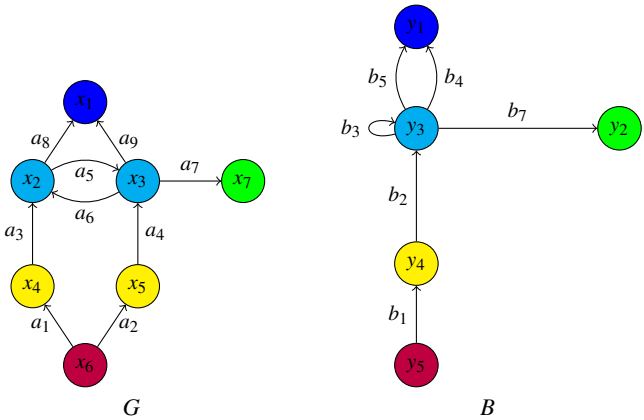

This fact corresponds to intuition: from the viewpoint of incoming signals there is no difference between $x_2$ and $x_3$: they are
bound to behave alike, because their view of the system is the same (in fact, the same as the view of $y_3$ in $B$).

Even in the case of symmetric graphs (graphs that can be thought of as undirected) fibrations generalize automorphisms:
for example, any $k$-regular graph  (i.e., any symmetric graph whose nodes have exactly $k$ incoming arcs paired with $k$ reverse outgoing arcs) 
with $n$ nodes can be fibred over a $k$-bouquet (a one-node graph with $k$ loops): in fact, there are $k!^n$ different ways to do that (any in-neighborhood
can be fibred using any bijection). Nonetheless, there exist rigid $k$-regular graphs---if $k\geq 3$, almost all $k$-regular graphs are
rigid!
So fibrations do describe a \emph{proper generalization} of automorphisms.

 \section{Universal fibrations and minimal fibrations}
\label{sec:ufc}

To proceed with our discussion, let us define in-trees.
An \emph{in-tree} is a graph with a selected node $r$, the root, and such that
every other node has exactly one directed path to the root; if $t$ is a node of
an in-tree, we sometimes use $t \to r$ for denoting the unique path from $t$ to
the root. If $T$ is an in-tree, we write $h(T)$ for its \emph{height} (the
length of the longest path). Finally, we write $T\upharpoonright k$ for the tree $T$
truncated at height $k$, that is, we eliminate all nodes at distance greater than
$k$ from the root.  
% A \emph{(symmetric) tree} is a (symmetric) graph with a
% selected node, the root, such that there is exactly one nonstuttering walk
% (path) from any node to the root: the notions of height and truncation carry on
% to this case. 
Unless otherwise stated, morphisms between trees are required to
preserve the root.

\paragraph{Input tree or universal total graphs (a.k.a.~views).}
Now given a graph $G$ and one of its nodes $x$, we define a possibly infinite in-tree $\view Gx$ called
the \emph{universal total graph of $G$ at $x$} (or, simply, the ``view of $x$ in $G$'') as follows:
\begin{itemize}
   \item the nodes of $\view Gx$ are the finite paths
   of $G$ ending in $x$;
   \item there is an arc from the node $\pi$ to the node $\pi'$ iff $\pi=a
   \pi'$ for some arc $a$ (if $G$ is coloured, then the arc gets the same
   colour as $a$).
\end{itemize}
The view is essentially an unrolling of all the paths of $G$ ending at $x$. Figure~\ref{fig:views} shows
the views of all the nodes of the graph $G$ in Figure~\ref{fig:rigid-example} (left). 

\begin{figure}
\centering
    \begin{tabular}{c|c}
    \hline
    $\view G{x_1}$ &
    \scalebox{.3}{
        \begin{tikzpicture}
        \tikzset{vertex/.style = {shape=circle,draw}}
        \tikzset{edge/.style = {->,> = latex'}}
        % vertices
        \node[] () at (6,6.3) {}; % space at the top
        \node[vertex] (v) at (8,6) {};
        \node[vertex] (v3) at  (4,5) {};
        \node[vertex] (v23) at  (2,4) {};
        \node[vertex] (v53) at  (7,4) {};
        \node[vertex] (v423) at  (0,3) {};
        \node[vertex] (v323) at  (3,3) {};
        \node[vertex] (v653) at  (7,3) {};
        \node[vertex] (v6423) at  (0,2) {};
        \node[vertex] (v2323) at  (2,2) {};
        \node[vertex] (v5323) at  (5,2) {};
        \node[vertex] (v42323) at  (1,1) {};        
        \node[vertex] (v32323) at  (3,1) {};        
        \node[vertex] (v65323) at  (5,1) {};        
        \node (vd42323) at  (1,0) {$\vdots$};        
        \node (vd32323) at  (3,0) {$\vdots$};    
        \node[vertex] (v2) at  (12,5) {};
        \node[vertex] (v32) at  (10,4) {};
        \node[vertex] (v42) at  (15,4) {};
        \node[vertex] (v532) at  (8,3) {};
        \node[vertex] (v232) at  (11,3) {};
        \node[vertex] (v642) at  (15,3) {};
        \node[vertex] (v6532) at  (8,2) {};
        \node[vertex] (v3232) at  (10,2) {};
        \node[vertex] (v4232) at  (13,2) {};
        \node[vertex] (v53232) at  (9,1) {};  
        \node[vertex] (v23232) at  (11,1) {};        
        \node[vertex] (v64232) at  (13,1) {};        
        \node (vd53232) at  (9,0) {$\vdots$};        
        \node (vd23232) at  (11,0) {$\vdots$};                
        %edges
        \draw[edge] (v2) to (v);
        \draw[edge] (v3) to (v);
        \draw[edge] (v23) to (v3);
        \draw[edge] (v53) to (v3);        
        \draw[edge] (v423) to (v23);
        \draw[edge] (v323) to (v23);
        \draw[edge] (v653) to (v53);
        \draw[edge] (v6423) to (v423);
        \draw[edge] (v2323) to (v323);
        \draw[edge] (v5323) to (v323);
        \draw[edge] (v42323) to (v2323);
        \draw[edge] (v32323) to (v2323);
        \draw[edge] (v65323) to (v5323);
        \draw[edge] (vd42323) to  (v42323);
        \draw[edge] (vd32323) to  (v32323);
        \draw[edge] (v32) to (v2);
        \draw[edge] (v42) to (v2);        
        \draw[edge] (v532) to (v32);
        \draw[edge] (v232) to (v32);
        \draw[edge] (v642) to (v42);
        \draw[edge] (v6532) to (v532);
        \draw[edge] (v3232) to (v232);
        \draw[edge] (v4232) to (v232);
        \draw[edge] (v53232) to (v3232);
        \draw[edge] (v23232) to (v3232);
        \draw[edge] (v64232) to (v4232);
        \draw[edge] (vd53232) to  (v53232);
        \draw[edge] (vd23232) to  (v23232);
        \end{tikzpicture}}
    \\
    \hline
    $\view G{x_2} \simeq \view G{x_3}$ &
    \scalebox{.3}{
        \begin{tikzpicture}
        \tikzset{vertex/.style = {shape=circle,draw}}
        \tikzset{edge/.style = {->,> = latex'}}
        % vertices
        \node[] () at (4,5.3) {}; % space at the top
        \node[vertex] (v3) at  (4,5) {};
        \node[vertex] (v23) at  (2,4) {};
        \node[vertex] (v53) at  (7,4) {};
        \node[vertex] (v423) at  (0,3) {};
        \node[vertex] (v323) at  (3,3) {};
        \node[vertex] (v653) at  (7,3) {};
        \node[vertex] (v6423) at  (0,2) {};
        \node[vertex] (v2323) at  (2,2) {};
        \node[vertex] (v5323) at  (5,2) {};
        \node[vertex] (v42323) at  (1,1) {};        
        \node[vertex] (v32323) at  (3,1) {};        
        \node[vertex] (v65323) at  (5,1) {};        
        \node (vd42323) at  (1,0) {$\vdots$};        
        \node (vd32323) at  (3,0) {$\vdots$};                
        %edges
        \draw[edge] (v23) to (v3);
        \draw[edge] (v53) to (v3);        
        \draw[edge] (v423) to (v23);
        \draw[edge] (v323) to (v23);
        \draw[edge] (v653) to (v53);
        \draw[edge] (v6423) to (v423);
        \draw[edge] (v2323) to (v323);
        \draw[edge] (v5323) to (v323);
        \draw[edge] (v42323) to (v2323);
        \draw[edge] (v32323) to (v2323);
        \draw[edge] (v65323) to (v5323);
        \draw[edge] (vd42323) to  (v42323);
        \draw[edge] (vd32323) to  (v32323);
        \end{tikzpicture}
    }
    \\
    \hline
    $\view G{x_4} \simeq \view G{x_5}$ &
    \scalebox{.3}{
        \begin{tikzpicture}
        \tikzset{vertex/.style = {shape=circle,draw}}
        \tikzset{edge/.style = {->,> = latex'}}
        % vertices
        \node[] () at (0,1.3) {}; % space at the top
        \node[vertex] (v4) at  (0,1) {};
        \node[vertex] (v64) at  (0,0) {};
        %edges
        \draw[edge] (v64) to (v4);
        \end{tikzpicture}
    }
    \\
        \hline
    $\view G{x_6}$ &
    \scalebox{.3}{
        \begin{tikzpicture}
        \tikzset{vertex/.style = {shape=circle,draw}}
        \tikzset{edge/.style = {->,> = latex'}}
        % vertices
        \node[] () at (0,0.3) {}; % space at the top
        \node[vertex] (v6) at  (0,0) {};
        \end{tikzpicture}
    }
    \\
    \hline
    $\view G{x_7}$ &
    \scalebox{.3}{
        \begin{tikzpicture}
        \tikzset{vertex/.style = {shape=circle,draw}}
        \tikzset{edge/.style = {->,> = latex'}}
        % vertices
        \node[] () at (4,6.3) {}; % space at the top
        \node[vertex] (v7) at (4,6) {};
        \node[vertex] (v3) at  (4,5) {};
        \node[vertex] (v23) at  (2,4) {};
        \node[vertex] (v53) at  (7,4) {};
        \node[vertex] (v423) at  (0,3) {};
        \node[vertex] (v323) at  (3,3) {};
        \node[vertex] (v653) at  (7,3) {};
        \node[vertex] (v6423) at  (0,2) {};
        \node[vertex] (v2323) at  (2,2) {};
        \node[vertex] (v5323) at  (5,2) {};
        \node[vertex] (v42323) at  (1,1) {};        
        \node[vertex] (v32323) at  (3,1) {};        
        \node[vertex] (v65323) at  (5,1) {};        
        \node (vd42323) at  (1,0) {$\vdots$};        
        \node (vd32323) at  (3,0) {$\vdots$};                
        %edges
        \draw[edge] (v3) to (v7);
        \draw[edge] (v23) to (v3);
        \draw[edge] (v53) to (v3);        
        \draw[edge] (v423) to (v23);
        \draw[edge] (v323) to (v23);
        \draw[edge] (v653) to (v53);
        \draw[edge] (v6423) to (v423);
        \draw[edge] (v2323) to (v323);
        \draw[edge] (v5323) to (v323);
        \draw[edge] (v42323) to (v2323);
        \draw[edge] (v32323) to (v2323);
        \draw[edge] (v65323) to (v5323);
        \draw[edge] (vd42323) to  (v42323);
        \draw[edge] (vd32323) to  (v32323);
        \end{tikzpicture}
    }
    \\\hline
    \end{tabular}
\caption{\label{fig:views}The views of all nodes of the graph $G$ in Figure~\ref{fig:rigid-example} (left), grouped when
they are isomorphic. Note that almost all of them are infinite in-trees (except for three nodes, that have only finitely many
incoming paths).}
\end{figure}
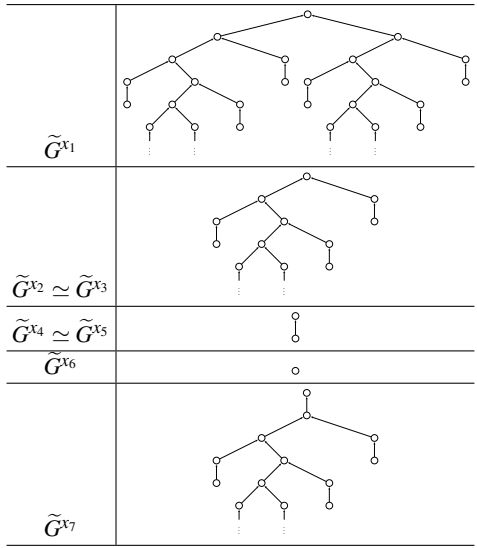

\paragraph{Nodes with the same view.}
We shall be interested in identifying the nodes of a graph $G$ sharing isomorphic
views. At this point, the reader may wonder whether and how this equivalence relation
between nodes can be made effective.
  
Let us start by providing a more humble equivalence relations
$\simeq_k$ on the nodes, defined by letting $x\simeq_k y$ iff $\view
Gx\upharpoonright k\cong\view Gy\upharpoonright k$, that is,
isomorphic only up to a level $k$ in the input tree.

\begin{teo}[See~\cite{NorUCG}]
\label{teo:norris}
  If $G$ has $n$ nodes, for all nodes $x,y$, $\view Gx\cong\view Gy$ iff $\view
  Gx\upharpoonright (n-1)\cong\view Gy\upharpoonright(n-1)$, that is, iff there is an isomorphism between the
  first $n-1$ levels of the two trees. 
\end{teo}

\paragraph{Minimum base and minimal fibrations.}
It is worth noticing that every fibration of a graph ``smashes together'' some
nodes that possess the same universal total graph. Formally stated, if 
$\phi:G \to B$ is a fibration, and $\phi(x)=\phi(y)$, then $\view Gx\cong
\view Gy$.

It is natural to ask whether it is possible to take this process to
extremes and identify any two nodes having the same universal total
graph. Not only this is possible, but what we obtain is the unique (up
to isomorphism) smallest base over which $G$ can be fibred.  The
fibration itself is also uniquely defined (up to composition with an
isomorphism) on the nodes, whereas not necessarily so on the
arcs (\cite{BoVGF}, Theorem 30). Formally,
we can summarize the key properties we need that were obtained in~\cite{BoVGF} as follows:

\begin{teo}
  A graph $G$ is \emph{fibration prime} iff every surjective fibration $G \to H$ is an isomorphism.
  For every graph $G$, there exists a graph $\hat G$ (called the \emph{minimum base of $G$}) such that:
  \begin{itemize}
    \item $\hat G$ is fibration prime;
    \item there is a surjective fibration $\mu: G \to \hat G$;
    \item if $G$ can be surjectively fibred over a fibration prime
      graph $B$, then $B\cong \hat G$; moreover if $\iota: B \to \hat
      G$ is an isomorphism, then for all $x \in N_G$,
      $\mu(x)=\iota(\phi(x))$.
  \end{itemize} 
\end{teo}
There is a set-partition algorithm that computes the minimum base
$\hat G$ of $G$ and a minimal fibration $\mu: G \to \hat G$.  It partitions the
graph into classes of nodes having the same universal total graph. To
do so, it works in $|N_G|-1$ rounds: the partition associated to round
$k$ is the one induced by the equivalence relation $\simeq_k$ (i.e.,
two nodes are in the same class iff they share the first $k$ levels of
their universal total graphs). By Theorem~\ref{teo:norris},  
after the last round two nodes are in the same class iff they
have the same universal total graphs.

At round $0$, all nodes are in the same class, because $\simeq_0$ is
the total relation. To build $\simeq_{k+1}$, just note that $x
\simeq_{k+1} y$ iff $x \simeq_k y$ and there is a bijection
$\psi:G(-,x)\to G(-,y)$ such that $s(a)\simeq_k s(b)$ and
$t(a)\simeq_k t(b)$. 

If we think of classes as a node colours, the overall process defines
a sequence of node colourings.
In the colouring at round $0$, all nodes have the same colour; 
at round $k+1$, we assign
the same colour to two nodes $x$ and $y$ iff the colours  of in-neighbours of $x$ and
$y$ at the $k$-th round were the same, and with the same multiplicity (i.e., for every colour $c$, $x$ has as many in-neighbours of
colour $c$ at round $k$ as $y$).

In fact, the graph $B$ in Figure~\ref{fig:rigid-example} is the
minimum base $\hat G$ of $G$, and it is built conceptually by
identifying the nodes of $G$ that have isomorphics views (see
Figure~\ref{fig:views}).

\medskip
From now on, we use $\mineq_G$ to denote the equivalence relation between nodes induced (according to Theorem~\ref{teo:lip}) by (any) 
minimal fibration $G \to \hat G$. In other words, we write $x\mineq_G y$ to mean that $\view Gx\cong\view Gy$ or that, equivalently,
$x$ and $y$ are mapped to the same node of the minimum base. 
The equivalence relation $\mineq_G$ is the coarsest equivalence relation among the nodes of $G$ satisfying the local in-isomorphism property.
The problem of finding $\mineq_G$ is known in the literature as the \emph{color refinement} (or \emph{naive vertex classification}) problem~\cite{berkholz2017tight},
and is much studied because it is a fundamental block in every efficient graph isomorphism procedure.
Most efficient techniques to compute $\mineq_G$ are generalizations of Hopcroft's finite-state automata minimization algorithm~\cite{hopcroft1971n};
in particular, the first $O(|A|\log |N|)$ algorithm to compute $\mineq_G$ is due to Cardon and Crochemore~\cite{cardon1982} (later
simplified in~\cite{paige1987three} and applied to information-processing networks in~\cite{monteiroAlgorithm}), and for this reason we shall often refer to $\mineq_G$ as the Cardon-Crochemore equivalence of $G$.

\section{Unveiling and Reconstructing Symmetries}
\label{sec:tools}

\subsection{Overall Plan}

We are now ready to present and discuss the objective of the paper: we
have a known graph $G$ that is a noisy version of an unknown graph
$H$; $G$ and $H$ have the same nodes, but $G$ may contain some (a
few) extra arcs, and may lack some (a few) of the original arcs.  In
principle, we would like to recover $H$, although the process will
unavoidably be approximated.  

The starting assumption of the whole reconstruction is that the unknown graph $H$ is rich of symmetries, 
in the sense that $\mineq_H$ is nontrivial and has some large equivalence classes.
Unfortunately, $G$ is not as symmetrical, and $\mineq_G$ is in general much less representative of the hidden symmetries in the network. We will therefore proceed as follows:

\begin{enumerate}
  \item[] {\bf Step (1).} [Algorithm~\ref{algo:equiv-rel}] We will first find an equivalence relation $\sim$ on the nodes of $G$ that is as close as possible to $\mineq_H$; normally, $H$ is unknown,
and so is $\mineq_H$. Our baseline will be $\mineq_G$, that is, the fibration symmetry that remains in $G$ despite of the noise, but 
our approach will try to do a better job and reconstruct more symmetries that those that are left in $G$.
  \item[] {\bf Step (2).} [Algorithm~\ref{algo:build-qf}] Based on
    $\sim$, we shall find a base graph $B$ and a homomorphism $\xi: G
    \to B$ whose fibres are exactly the equivalence classes of
    $\sim$. The homomorphism $\xi$ will be surjective, but not a
    fibration because of the presence of noise. Nonetheless it will be
    ``close to'' a fibration, in a technical sense that we will make
    more precise below. We will call $\xi$ a quasifibration.
  \item[] {\bf Step (3).} [Algorithm~\ref{algo:qf-to-fib}] Now we will turn the quasifibration $\xi: G \to B$ into an actual fibration $\xi': G' \to B$, where $G'$ is a minimally modified version
  of $G$ that allows to transform $\xi$ into a fibration.
  \item[] {\bf Step (4).} We compute the minimum base of $G'$ and a minimal fibration $\mu: G' \to \hat{G'}$; the equivalence relation $\mineq_{G'}$
  will provide the final fibration symmetry we are looking for. It is a coarsening of $\mineq_{G'}=\sim$.
\end{enumerate}

The final result will be the reconstructed ``symmetrical'' version $G'$ of $G$ (that aims at being $H$) along with a fibration $\xi': G' \to B$
that explains the symmetries.
 
The three steps above will be totally orthogonal: in particular, as we
will see, steps (2), (3) and (4) above will be performed in a provably
optimal way. However, determining $\sim$ in an optimal way (that is,
deriving $\mineq_H$ from $G$) is not generally possible, and depends
largely from how much noise was introduced and where. As a
consequence, we will approach step (1) based on some heuristics in
Algorithm 3 that we will discuss at the right time.
Since the first three steps above are independent of one another, we can describe them in any order. It will actually be easier to discuss them in 
reverse order, starting from the last and moving to the first.

\subsection{Step (4): Building the final fibration symmetry}

The role of step (4) may be unclear to the reader: the point is that, once we have de-noised $G$ obtaining $G'$, we want to build its
\emph{full} fibration symmetry, constructing a minimal fibration. Of course, the equivalence relation $\sim$ on the nodes of $G$ found
in step (1) satisfies the local in-isomorphism property on $G'$ (steps (2) and (3) serve precisely the purpose of adjusting
$G'$ so that $\sim$ becomes a local in-isomorphism), but it is not the coarsest one. That is why we finally compute the coarsest local
in-isomorphism on the nodes $G'$: of course, $\mineq_{G'}$ will  be a coarsening of $\sim$.

\subsection{Step (3): Quasifibrations and the {\scshape GraphRepair} algorithm}

Before discussing step (3), let us provide as promised a relaxed
version of the definition of fibration.  Recall that a fibration
requires that every arc of the base can be lifted to exactly one
incoming arc at every node of the fibre of its target.  We can easily
derive a measure of how close (or far) a given graph homomorphism is
from being a fibration.

For a given graph homomorphism $\xi: G \to B$, consider an arc $a\in A_B$
and a node $y \in N_G$ such that $\xi(y)=t(a)$: we may wonder if there
are arcs of $G$, with target $y$, that are image of $a$; in the light of Definition~\ref{def:fib}, these 
are all potential ``liftings of $a$ at $y$''. If $\xi$ is a fibration,
there must be exactly one lifting (because the definition of fibration entails existence
and uniqueness). Every time this set of arcs is empty (i.e., existence fails) or
contains more than one element (i.e., uniqueness fails), we have a witness of the fact that $\xi$ is not a fibration.
Formally:

 \begin{defi}
 \label{def:quasi-fib}
Given a graph homomorphism $\xi: G \to B$, for every arc $a\in A_B$
and every node $y \in N_G$ such that $\xi(y)=t(a)$, let us define the
\emph{local delta function} as the number of liftings minus 1:
\[
	\delta_\xi(a, y) = \left|\set{a' \in A_G \mid \xi(a')=a \text{ and } t(a')=y}\right|-1.
\]
We let
\begin{eqnarray*}
	\Delta_\xi^+ &=& \sum_{a,y} \max(\delta_\xi(a, y), 0) \quad\text{the \emph{excess} of $\xi$}\\ 
	\Delta_\xi^- &=& \sum_{a,y} \max(-\delta_\xi(a, h), 0) \quad\text{the \emph{deficiency} of $\xi$}\\
	\Delta_\xi 
	 &=& \sum_{a,y} \left|\delta_\xi(a,y)\right| =\\
	 &=& \Delta_\xi^+ + \Delta_\xi^- \quad\text{the \emph{total error} of $\xi$,}
\end{eqnarray*}
where all summations range over all  possible pairs $a \in A_B$ and $y \in \xi^{-1}(t_B(a))$.
We say that $\xi$ is a \emph{$(\Delta_\xi^-,\Delta_\xi^+)$-quasifibration} or just a $\Delta_\xi$-quasifibration.
 \end{defi}

If $\xi$ is a fibration, the function $\delta_\xi$ is constantly zero
(because every arc can be lifted in a unique way).  For a homomorphism
that fails to be a fibration, $\delta_\xi(a,y)=-1$ when $a$ cannot be
lifted at $x$; $\delta_\xi(a,y)>1$ when $a$ can be lifted in more than
one way at $y$.

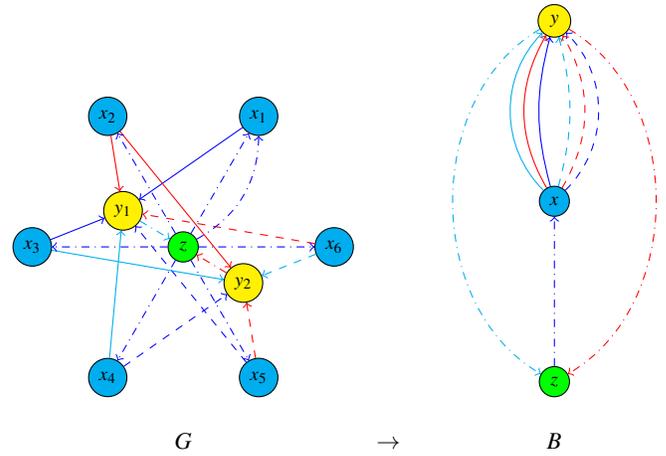
\begin{figure}
\centering
\begin{tabular}{ccc}
\begin{tikzpicture}[scale=0.8, every node/.style={scale=0.8}]
	\node[regular polygon, regular polygon sides=6, minimum size=5cm] at (0,0) (A) {};
	\foreach \i in {1,...,6}
		\node[circle, draw, fill=cyan] at (A.corner \i) (X\i) {$x_{\i}$};
	\node[circle, draw, fill=yellow] at (-1,0.6) (L) {$y_1$};
	\node[circle, draw, fill=yellow] at (1,-0.6) (R) {$y_2$};
	\node[circle, draw, fill=green] at (0,0) (C) {$z$};

	\draw[->,blue] (X1) to (L);                                              % a1: x1 -> y1     (a1)
	\draw[->,red] (X2) to (L);                                               % a2: x2 -> y1     (a2)
	%\draw[->,cyan,dashed] (X3) to (L);                                      % a3: x3 -> y1     (a3)     MISSING (a3, y1)
	\draw[->,blue] (X3) to (L); %%%                                          % extra2: x3 -> y1 (a1)     EXTRA (a1,y1) -> a1, extra2
	\draw[->,cyan] (X4) to (L);                                              % a4: x4 -> y1     (a4)
	\draw[->,blue,dashed] (X5) to (L);                                       % a5: x5 -> y1     (a5)
	\draw[->,red,dashed] (X6) to (L);                                        % a6: x6 -> y1     (a6)

	%\draw[->,blue] (X1) to (R);                                             % b1: x1 -> y2     (a1)     MISSING (a1, y2)
	\draw[->,red] (X2) to (R);                                               % b2: x2 -> y2     (a2)
	\draw[->,cyan] (X3) to (R);                                              % b3: x3 -> y2     (a3)
	\draw[->,blue,dashed] (X4) to (R);                                       % b4: x4 -> y2     (a4)
	\draw[->,red,dashed] (X5) to (R);                                        % b5: x5 -> y2     (a5)
	\draw[->,cyan,dashed] (X6) to (R);                                       % b6: x6 -> y2     (a6)

	\draw[->,cyan,dashdotted] (L) to (C);                   % c1: y1 -> z	    (c1)
	\draw[->,red,dashdotted] (R) to (C);	                 % c2: y2 -> z      (c2)

	\draw[->,blue,dashdotted] (C) to (X1);				     % d1: z -> x1      (b)      EXTRA (b,x1) -> d1, extra1
	\draw[->,blue,dashdotted] (C) to[bend right] (X1); %%%  % extra1: z -> x1  (b)
	\draw[->,blue,dashdotted] (C) to (X2);				     % d2: z -> x2      (b)
	\draw[->,blue,dashdotted] (C) to (X3);				     % d3: z -> x3      (b)
	\draw[->,blue,dashdotted] (C) to (X4);				     % d4: z -> x4      (b)
	\draw[->,blue,dashdotted] (C) to (X5);				     % d5: z -> x5      (b)
	\draw[->,blue,dashdotted] (C) to (X6);				     % d6: z -> x6      (b)
\end{tikzpicture}
& $\qquad$ &
\begin{tikzpicture}[scale=0.8, every node/.style={scale=0.8}]
	\node[circle, draw, fill=green]  at (0,0) (Z) {$z$};
	\node[circle, draw, fill=cyan] at (0,3) (X) {$x$};
	\node[circle, draw, fill=yellow] at (0,6) (Y) {$y$};
	
	\draw[->,blue] (X) to[bend left=15] (Y);                                 % a1: x -> y	                             
	\draw[->,red] (X) to[bend left=30] (Y);	                                 % a2: x -> y
	\draw[->,cyan,dashed] (X) to[bend left=-15] (Y);	                     % a3: x -> y
	\draw[->,cyan] (X) to[bend left=45] (Y);                                 % a4: x -> y	
	\draw[->,blue,dashed] (X) to[bend left=-45] (Y);	                     % a5: x -> y
	\draw[->,red,dashed] (X) to[bend left=-30] (Y);	                         % a6: x -> y
	
	\draw[->,cyan,dashdotted] (Y) to[bend right=60] (Z);    % c1: y -> z
	\draw[->,red,dashdotted] (Y) to[bend left=60] (Z);	     % c2: y -> z
	
	\draw[->,blue,dashdotted] (Z) to (X);                   % b: z -> x
\end{tikzpicture} \\[1em]
$G$ & $\to$ & $B$
\end{tabular}
\caption{\label{fig:qf-example}An example of a quasifibration $\xi:
  G\to B$. Node colour and arc shape and colour are used to suggest
  how nodes and arcs are mapped. It is easy to check that this is a
  graph homomorphism.}
\end{figure}

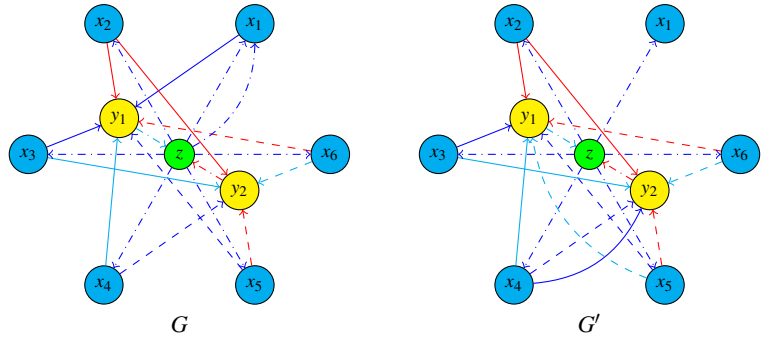
\begin{figure*}[ht!]
    \centering
    \begin{tabular}{ccc}
\begin{tikzpicture}[scale=0.8, every node/.style={scale=0.8}]
	\node[regular polygon, regular polygon sides=6, minimum size=5cm] at (0,0) (A) {};
	\foreach \i in {1,...,6}
		\node[circle, draw, fill=cyan] at (A.corner \i) (X\i) {$x_{\i}$};
	\node[circle, draw, fill=yellow] at (-1,0.6) (L) {$y_1$};
	\node[circle, draw, fill=yellow] at (1,-0.6) (R) {$y_2$};
	\node[circle, draw, fill=green] at (0,0) (C) {$z$};

	\draw[->,blue] (X1) to (L);                                              % a1: x1 -> y1     (a1)
	\draw[->,red] (X2) to (L);                                               % a2: x2 -> y1     (a2)
	%\draw[->,cyan,dashed] (X3) to (L);                                      % a3: x3 -> y1     (a3)     MISSING (a3, y1)
	\draw[->,blue] (X3) to (L); %%%                                          % extra2: x3 -> y1 (a1)     EXTRA (a1,y1) -> a1, extra2
	\draw[->,cyan] (X4) to (L);                                              % a4: x4 -> y1     (a4)
	\draw[->,blue,dashed] (X5) to (L);                                       % a5: x5 -> y1     (a5)
	\draw[->,red,dashed] (X6) to (L);                                        % a6: x6 -> y1     (a6)

	%\draw[->,blue] (X1) to (R);                                             % b1: x1 -> y2     (a1)     MISSING (a1, y2)
	\draw[->,red] (X2) to (R);                                               % b2: x2 -> y2     (a2)
	\draw[->,cyan] (X3) to (R);                                              % b3: x3 -> y2     (a3)
	\draw[->,blue,dashed] (X4) to (R);                                       % b4: x4 -> y2     (a4)
	\draw[->,red,dashed] (X5) to (R);                                        % b5: x5 -> y2     (a5)
	\draw[->,cyan,dashed] (X6) to (R);                                       % b6: x6 -> y2     (a6)

	\draw[->,cyan,dashdotted] (L) to (C);                   % c1: y1 -> z	    (c1)
	\draw[->,red,dashdotted] (R) to (C);	                 % c2: y2 -> z      (c2)

	\draw[->,blue,dashdotted] (C) to (X1);				     % d1: z -> x1      (b)      EXTRA (b,x1) -> d1, extra1
	\draw[->,blue,dashdotted] (C) to[bend right] (X1); %%%  % extra1: z -> x1  (b)
	\draw[->,blue,dashdotted] (C) to (X2);				     % d2: z -> x2      (b)
	\draw[->,blue,dashdotted] (C) to (X3);				     % d3: z -> x3      (b)
	\draw[->,blue,dashdotted] (C) to (X4);				     % d4: z -> x4      (b)
	\draw[->,blue,dashdotted] (C) to (X5);				     % d5: z -> x5      (b)
	\draw[->,blue,dashdotted] (C) to (X6);				     % d6: z -> x6      (b)
\end{tikzpicture}
& $\qquad$ &
\begin{tikzpicture}[scale=0.8, every node/.style={scale=0.8}]
	\node[regular polygon, regular polygon sides=6, minimum size=5cm] at (0,0) (A) {};
	\foreach \i in {1,...,6}
		\node[circle, draw, fill=cyan] at (A.corner \i) (X\i) {$x_{\i}$};
	\node[circle, draw, fill=yellow] at (-1,0.6) (L) {$y_1$};
	\node[circle, draw, fill=yellow] at (1,-0.6) (R) {$y_2$};
	\node[circle, draw, fill=green] at (0,0) (C) {$z$};

	%%%%%%\draw[->,blue] (X1) to (L);                                        % a1: x1 -> y1     (a1)            FIXING REMOVED
	\draw[->,red] (X2) to (L);                                               % a2: x2 -> y1     (a2)
	%\draw[->,cyan,dashed] (X3) to (L);                                      % a3: x3 -> y1     (a3)     MISSING (a3, y1)
	\draw[->,blue] (X3) to (L); %%%                                          % extra2: x3 -> y1 (a1)     EXTRA (a1,y1) -> a1, extra2
	\draw[->,cyan] (X4) to (L);                                              % a4: x4 -> y1     (a4)
	\draw[->,blue,dashed] (X5) to (L);                                       % a5: x5 -> y1     (a5)
	\draw[->,cyan,dashed] (X5) to[bend left] (L);                                       % new_arc_265: x5 -> y1 (a3)       FIXING ADDED
	\draw[->,red,dashed] (X6) to (L);                                        % a6: x6 -> y1     (a6)

	%\draw[->,blue] (X1) to (R);                                             % b1: x1 -> y2     (a1)     MISSING (a1, y2)
	\draw[->,blue] (X4) to[bend right] (R);    %%%%%%                                    % new_arc_864: x4 -> y2 (a1)       FIXING ADDED
	\draw[->,red] (X2) to (R);                                               % b2: x2 -> y2     (a2)
	\draw[->,cyan] (X3) to (R);                                              % b3: x3 -> y2     (a3)
	\draw[->,blue,dashed] (X4) to (R);                                       % b4: x4 -> y2     (a4)
	\draw[->,red,dashed] (X5) to (R);                                        % b5: x5 -> y2     (a5)
	\draw[->,cyan,dashed] (X6) to (R);                                       % b6: x6 -> y2     (a6)

	\draw[->,cyan,dashdotted] (L) to (C);                   % c1: y1 -> z	    (c1)
	\draw[->,red,dashdotted] (R) to (C);	                 % c2: y2 -> z      (c2)

	%%%%%%\draw[->,blue,dashdotted] (C) to (X1);			  % d1: z -> x1      (b)      EXTRA (b,x1) -> d1, extra1   FIXING REMOVED
	\draw[->,blue,dashdotted] (C) to (X1); %%%  % extra1: z -> x1  (b) 
	\draw[->,blue,dashdotted] (C) to (X2);				     % d2: z -> x2      (b)
	\draw[->,blue,dashdotted] (C) to (X3);				     % d3: z -> x3      (b)
	\draw[->,blue,dashdotted] (C) to (X4);				     % d4: z -> x4      (b)
	\draw[->,blue,dashdotted] (C) to (X5);				     % d5: z -> x5      (b)
	\draw[->,blue,dashdotted] (C) to (X6);				     % d6: z -> x6      (b)
\end{tikzpicture}\\
$G$ &\qquad & $G'$
\end{tabular}
    \caption{\label{fig:qf-example-fixed}The original graph $G$ (total graph of the quasifibration $\xi: G \to B$) of Figure~\ref{fig:qf-example} 
    and the adjusted version $G'$ obtained after the application of Algorithm~\ref{algo:qf-to-fib}; the fibration $\xi': G' \to B$ is
    the one suggested by node and arc colouring, as usual. The blue arc $x_1\to y_1$ and the snake blue
    arc $z \to x_1$ were removed. A new blue arc $x_2\to y_2$ and a new cyan dashed arc $x_5 \to y_1$ were added.} 
\end{figure*}

We look at a quasifibration $\xi: G \to B$ as a ``fibration with errors'', where the errors are due to the presence of noise in the arcs of the total graph $G$. In this sense, we may want to reconstruct the correct (de-noised) total graph and make $\xi$ into a fibration.

In order to do this, let us say that two graphs $G_1$ and $G_2$ are \emph{compatible} iff they have the same nodes (i.e., $N_{G_1}=N_{G_2}$) and common arcs have the same sources and targets (i.e., if $a \in A_{G_1} \cap A_{G_2}$ then $s_{G_1}(a)=s_{G_2}(a)$ and $t_{G_1}(a)=t_{G_2}(a)$). Given two compatible graphs, we write $G_1 \Delta G_2$ for the cardinality of the symmetric difference\footnote{$A \Delta B$ is the standard notation used for the symmetric difference of two sets (e.g., $\{\,1,2,3,4\,\} \Delta \{\,3,4,5\,\}=\{\,1,2,5\,\}$); since in this paper we only need its cardinality, we make an abuse of notation and use $A \Delta B$ for the cardinality of the difference.} $(A_{G_1}\setminus A_{G_2}) \cup (A_{G_2}\setminus A_{G_1})$.

\begin{algorithm}[H]
    \SetAlgoLined
    \LinesNumbered
    \SetKwInOut{Input}{Input}
    \SetKwInOut{Output}{Output}
    \Input{a quasifibration $\xi: G \to B$}
    \Output{a fibration $\xi': G' \to B$}
    $G' \leftarrow G$\;
    $\xi' \leftarrow \xi$\;
    
    \For{$a \in A_B$} {
        \For{$y \in \xi^{-1}(t_B(a))$} {
            $U_{a,y} \leftarrow \set{a' \in A_G \mid \xi(a')=a \text{ and } t(a')=y}$\;
            \eIf{$|U_{a,y}| \geq 1$}{
                remove from $A_{G'}$ all elements of $U_{a,y}$ except one\;
            }{
                add a new arc $a'$ to $G'$ with target $y$ and source chosen in $\xi^{-1}(s_B(a))$ \;
                $\xi'(a') \leftarrow a$\;
            }
        }
    }
    \Return{$\xi': G' \to B$}
    \caption{\label{algo:qf-to-fib}{\scshape GraphRepair}: This algorithm reconstructs a graph $G'$ and a fibration $\xi'$ from a quasifibration $\xi: G \to B$ with the properties described in Theorem~\ref{teo:qf-to-fib}.}
\end{algorithm}

\begin{theorem}
\label{teo:qf-to-fib}
Algorithm~\ref{algo:qf-to-fib}, given a surjective quasifibration $\xi: G \to B$, reconstructs a surjective fibration $\xi': G' \to B$ such that
\begin{enumerate}
  \item $G$ and $G'$ are compatible, and $G \Delta G'=\Delta_\xi$;
  \item $\xi'$ coincides with $\xi$ on all nodes and on common arcs;
  \item for any other graph $G''$ and fibration $\xi'': G'' \to B$ with the same two properties, the difference $G \Delta G''$ is
  at least $\Delta_\xi$.
\end{enumerate}
\end{theorem} 
\begin{proof}
Algorithm 1 precisely determines for every arc $a$ of the base graph
and every $x$ in the counterimage of its target, if the existence and
uniqueness required by the definition of fibration
(Definition~\ref{def:fib}) are satisfied.  If uniqueness does not
hold, excess arcs are removed; if existence does not hold, a new arc
is added.  Since the arcs removed or added are exactly $\Delta_\xi$,
items (1) and (2) are clearly true.  Assume that we have some $G''$
satisfying item (3) but with a smaller difference.  This can only mean
that $G''$ either has one less arc of those that we added, or one more
arc of those that we deleted. In both cases, $\xi''$ would not be a
fibration.
\end{proof}

The above theorem tells us that Algorithm~\ref{algo:qf-to-fib} adjusts a graph so to turn a quasifibration into a fibration, with a minimal number
of changes. Note that while it is impossible to do this with less changes, the construction of the new graph is somehow arbitrary, because we can choose
arbitrarily which excess arcs to remove and to which exact source to connect the new arcs that were added to correct deficiencies (the two
branches of the if in Algorithm~\ref{algo:qf-to-fib}).

Let us describe how this procedure is applied to the graph $G$ of Figure~\ref{fig:qf-example}:
\begin{itemize}
  \item Node $y \in N_B$ has six in-coming arcs (blue, red, cyan, dashed blue, dashed red, dashed cyan), all having $x$ as source. 
So, the in-neighborhood of $y_1$ in $G$ should contain exactly one instance of these six arcs of the base, each coming from a blue node.
We observe the following exceptions: there are two blue arcs (from $x_3$ and from $x_1$) instead of one, 
and no dashed cyan. Similarly, looking at the in-neighborhood of $y_2$, we find that the blue arc is missing.
  \item Node $x \in N_B$ has only one in-coming arc (the snake blue arc) coming from $z$. Looking at the $x_i$'s in $G$, we note
that $x_1$ has two incoming snake blue arc.   
\end{itemize}

No other exception is present. So $\xi$ has a total excess of 2 and a total deficiency of 2; the total error $\Delta_\xi$ is 4.
If we apply Algorithm~\ref{algo:qf-to-fib} to $G$ we might obtain the graph $G'$ shown in Figure~\ref{fig:qf-example-fixed}.
Note that this is \emph{not} the only graph we can obtain, because of the non-determinism of the algorithm: for instance,
we might as well have removed the blue arc $x_3 \to y_1$ (instead of $x_1 \to y_1$); similarly, the source of the new blue arc $x_4 \to y_2$
might have been chosen differently (for instance, we might have chosen it to be $x_1 \to y_2$). The same for the new cyan dashed
arc $x_5 \to y_1$. 

The reader may be unsatisfied with this result: (s)he might have expected to obtain the much more symmetric graph of Figure~\ref{fig:intuition-example-colored} (left).
Recall, however, that we are not looking for an automorphism symmetry, but rather for a(ny) fibration symmetry.
 
\subsection{Step (2): The {\scshape BuildQF} algorithm}

Step (2) of our plan requires building a quasifibration from an
equivalence relation on the nodes.  More precisely, given an
equivalence relation $\sim$ on the nodes of a graph $G$, we want to
determine a base graph $B$ and an surjective homomorphism $\xi: G \to B$ whose
fibres are exactly the equivalence classes of $\sim$, and having the
smallest total error (i.e., $\xi$ must be as close as possible to a
fibration).  This is what Algorithm~\ref{algo:build-qf} does, as
ensured by Theorem~\ref{teo:build-qf}.

\begin{algorithm}[H]
    \SetAlgoLined
    \LinesNumbered
    \SetKwInOut{Input}{Input}
    \SetKwInOut{Output}{Output}
    \Input{a graph $G$ with an equivalence relation $\sim$ on $N_G$}
    \Output{a quasifibration $\xi: G \to B$}
    $N_B\leftarrow N_G/\sim$ (the equivalence classes of $\sim$)\;
    $\xi(x)\leftarrow [x]_\sim$ (the equivalence class containing $x$)\;
    
    \For{every pair of classes $X,Y\in N_B$} {
        let $Y=\set{y_1,\dots,y_k}$\;
        \For{$i=1,\dots,k$} {
            $v_i \leftarrow \text{number of arcs $a \in G(-,y_i)$ such that $s_G(a) \in X$}$\;
        }
        \If{all $v_i$ are $0$}{
            \textbf{continue}\;
        }
        let $z$ be (any) positive median of $\langle v_1,\dots,v_k\rangle$\;
        add to the graph $B$ new arcs $b_0,\dots,b_{z-1}$ with source $X$ and target $Y$\;
        \For{$i=1,\dots,k$} {
            $c\leftarrow 0$\;
            \For{$a \in G(-,y_i)$ such that $s(a) \in X$} {
                $\xi(a)\leftarrow b_{c \bmod z}$\;
                $c\leftarrow c+1$\;
            }
        }
    }
    \Return{$\xi: G \to B$}
    \caption{\label{algo:build-qf}{\scshape BuildQF}: This algorithm builds a quasifibration $\xi: G \to B$ with the properties described in Theorem~\ref{teo:build-qf}.}
\end{algorithm}

\begin{theorem}
\label{teo:build-qf}
Algorithm~\ref{algo:build-qf}, given a graph $G$ and an equivalence
relation $\sim$ on its nodes, builds an surjective quasifibration
$\xi: G \to B$ such that:
\begin{enumerate}
  \item the equivalence classes of $\sim$ are the fibres of $\xi$;
  \item if $\xi': G \to B'$ is a surjective quasifibration whose fibres are the equivalence classes of $\sim$, then the
  total error of $\xi'$ is at least as large as that of $\xi$.
\end{enumerate}
\end{theorem} 
\begin{proof}
Let us consider a generic surjective homomorphism $\xi: G \to B$ satisfying the
first property of the statement.  We can identify the nodes of $B$
with the equivalence classes of $\sim$ (because $\xi$ is surjective
and its fibres are required to be the equivalence classes of $\sim$).
Let us write $T\in \sim$ to mean that $T$ is an equivalence class, and
$[x]$ for the equivalence class of node $x \in N_G$.

With some abuse of notation, let $\tld ay$ denote the set of
$\xi$-counterimages of $a \in A_B$ with target $y \in N_G$ (this set
is a singleton if $\xi$ is a fibration, but it can have arbitrary
cardinality otherwise).  We can write the total error as
\[
	\Delta_\xi=\sum_{\substack{X,Y \in \sim\\Y=\set{y_1,\dots,y_k}}} \sum_{i=1}^k \sum_{a \in B(X,Y)} \left|\left|\tld a{y_i}\right|-1\right|.
\]
Now, consider the contribution of two specific $X,Y\in \sim$ in the
summation; let $Y=\set{y_1,\dots,y_k}$ and $v_i$ be the cardinality of
$|G(X,y_i)|=v_i$.  The function $\xi$ determines a matrix $M$ of
natural numbers, with $k$ rows and $h=|B(X,Y)|$ columns, where
$m_{ij}$ is the number of arcs in $G(X,y_i)$ that are mapped by $\xi$
to a specific arc of $B(X,Y)$.  The sum of the $i$-th row must equal $v_i$.
We want to choose $h$ and the matrix $M$ so that $\sum_{i,j}
|m_{ij}-1|$ is minimized.  The problem of finding this matrix can be
seen as an integer optimization problem with variables $m_{ij}$ (the
number of columns $h$ is also part of the optimization) where the
objective is
\[
	\text{minimize } \sum_{i=1}^k \sum_{j=1}^h|m_{ij}-1|
\]
with linear constraints:
\[
\begin{cases}
	&\sum_j m_{ij} = v_i \quad \text{i=1,\dots,k} \\
	& 1 \leq h\\
	&h, m_{ij} \in {\mathbf N}.
\end{cases}
\]
Since there is no interdependence between the rows, we can rewrite the system as follows: let $z_i$ be the number of null entries on the $i$-th
row; the remaining $h-z_i$ entries are $\geq 1$ and their sum is required to be $v_i$. So, for every row $i=1,\dots,k$
\[
	\sum_{j=1}^h|m_{ij}-1| = z_i + v_i - (h - z_i) = v_i - h + 2z_i. 
\]
We can thus equivalently state the optimization problem as follows\footnote{We can omit from the objective function the term $\sum_i v_i=|G(X,Y)|$ because
it is constant and does not influence minimization.}
\[
	\text{minimize }\ 2\left(\sum_{i=1}^k z_i\right)-hk
\]
with linear constraints:
\[
\begin{cases}
	&h-v_i \leq z_i \quad \text{i=1,\dots,k}\\
	&1 \leq h\\
	&z_i, h \in {\mathbf N},
\end{cases}
\]
because no row can contain more than $v_i$ non-null entries ($h-z_i\leq v_i$).
This is now an ILP problem with $k+1$ variables $(z_1,\dots,z_k,h)$. 

Let us assume without loss of generality that $v_1\leq v_2\leq \dots
\leq v_k$, and consider an optimal solution $(z_1^*,\dots,z_k^*,h^*)$.
Let $i^*$ be the largest index for which $h^*\geq v_{i^*}$.

Define a new solution $(\bar z_1,\dots,\bar z_k,h^*)$ by letting
\[
	\bar z_i=\max(h^*-v_i,0).
\]
 It is easy to see that this is also an admissible solution, and its cost is
\begin{multline*}
	2\sum_{i=1}^k \max(h^*-v_i,0)-h^*k=
	2\sum_{i=1}^{i^*} (h^*-v_i)-h^*k\\
	\leq 2 \sum_{i=1}^{i^*} z_i^*-h^*k 
	\leq 2 \sum_{i=1}^{k} z_i^*-h^*k;  
\end{multline*}
since $(z_1^*,\dots,z_k^*,h^*)$ is optimal, so is $(\bar z_1,\dots,\bar z_k,h^*)$.
The problem is then just to find $h^*$ minimizing
\[
	f(h)=2\sum_{i=1}^k \max(h-v_i,0)-hk.
\]
If $h<v_1$, any increase in $h$ reduces $f(h)$.
If $v_\ell\leq h\leq h+d\leq v_{\ell+1}$, then
\begin{multline*}
f(h+d)-f(h)=\\
=2\sum_{i=1}^k \max(h+d-v_i,0)-(h+d)k-2\sum_{i=1}^k \max(h-v_i,0)-hk=\\
=2\sum_{i=1}^\ell (h+d-v_i)-(h+d)k-2\sum_{i=1}^\ell (h-v_i)+hk=d(2\ell-k),
\end{multline*}
which is negative when $k<\ell/2$, and positive when $k>\ell/2$. Hence $h^*$ should be chosen as the median
of $v_1,\dots,v_k$.

These considerations explain why the choice made by Algorithm~\ref{algo:build-qf} is optimal. Note that we need to guarantee
that the median be $\geq 1$ otherwise the condition on row sums (for positive $v_i$) could not be satisfied.
As for the 
actual assignment of the arcs of $G(-,y_i)$ to the arcs of $B(X,Y)$ (the inner cycle in the Algorithm), 
it is irrelevant for the optimality.
\end{proof}

Fig.~\ref{fig:qf-build-example} shows an example of Application of Algorithm~\ref{algo:build-qf}:
\begin{itemize}
  \item ({\bf Left}) The graph $G$ (left) is ``almost'' symmetrical: adding the three missing arcs (dotted in the picture) and
removing the blue arc would produce the local in-isomorphism with the two classes represented by the colours.
  \item ({\bf Right top}) The table contains how many arcs come in every node of a specific class (row) coming from nodes of another class (column): for instance
the cyan/yellow entry means that three yellow nodes have 3 incoming arcs from cyan nodes, one ($y_5$) has two and another one ($y_1$) has one. 
  \item ({\bf Right bottom}) The resulting graph $B$: it has one node for each class ($x$ for the cyan class, $y$ for the yellow class),
and the number of arcs from class $X$ to class $Y$ is the median of the sequence in the table above.  The map $\xi: G \to B$ built
by the algorithm is a quasifibration with excess 1 and deficiency 4 (the map on the arcs is not shown, for the sake of simplicity). 
\end{itemize}

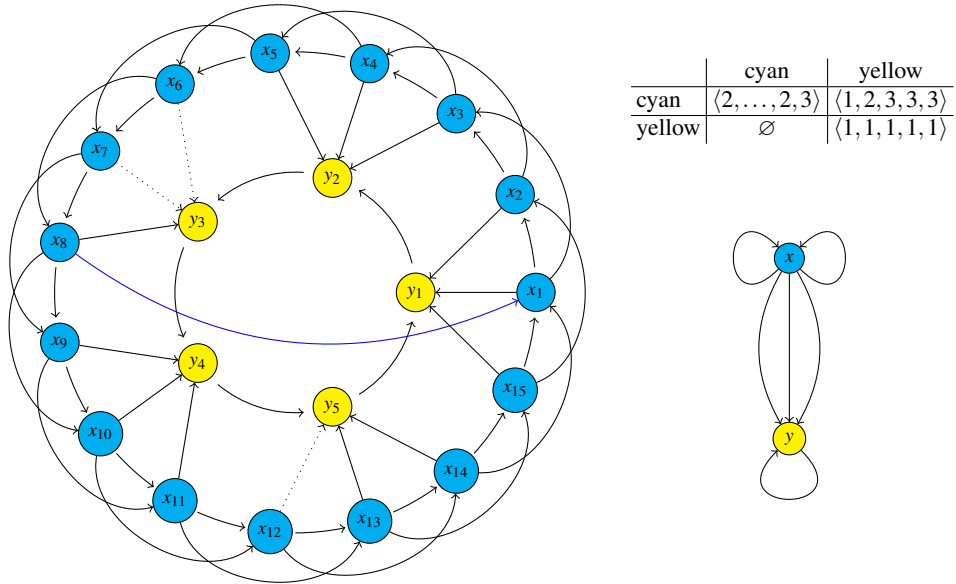
\begin{figure*}
\centering
\begin{tabular}{cc}
\raisebox{-.5\height}{\begin{tikzpicture}[scale=0.8, every node/.style={scale=0.8}]
\def \n {15}
\def \radius {4cm}
\def \margin {6} % margin in angles, depends on the radius
\foreach \s in {1,...,\n}
{
  \node[draw, circle, fill=cyan] (x\s) at ({360/\n * (\s - 1)}:\radius) {$x_{\s}$};
  \draw[->] ({360/\n * (\s - 1)+\margin}:\radius) 
    arc ({360/\n * (\s - 1)+\margin}:{360/\n * (\s)-\margin}:\radius); 
}
  \draw[->] (x1) to[bend right=70] (x3);
  \draw[->] (x2) to[bend right=70] (x4);
  \draw[->] (x3) to[bend right=70] (x5);
  \draw[->] (x4) to[bend right=70] (x6);
  \draw[->] (x5) to[bend right=70] (x7);
  \draw[->] (x6) to[bend right=70] (x8);
  \draw[->] (x7) to[bend right=70] (x9);
  \draw[->] (x8) to[bend right=70] (x10);
  \draw[->] (x9) to[bend right=70] (x11);
  \draw[->] (x10) to[bend right=70] (x12);
  \draw[->] (x11) to[bend right=70] (x13);
  \draw[->] (x12) to[bend right=70] (x14);
  \draw[->] (x13) to[bend right=70] (x15);
  \draw[->] (x14) to[bend right=70] (x1);
  \draw[->] (x15) to[bend right=70] (x2);
\def \n {5}
\def \radius {2cm}
\def \margin {14} % margin in angles, depends on the radius
\foreach \s in {1,...,\n}
{
  \node[draw, circle, fill=yellow] (y\s) at ({360/\n * (\s - 1)}:\radius) {$y_{\s}$};
  \draw[->] ({360/\n * (\s - 1)+\margin}:\radius) 
    arc ({360/\n * (\s - 1)+\margin}:{360/\n * (\s)-\margin}:\radius);
}

  \draw[->, blue] (x8) to[bend right=30] (x1);

  \draw[->, dotted] (x6) to (y3);
  \draw[->, dotted] (x7) to (y3);
  \draw[->] (x8) to (y3);

  \draw[->] (x9) to (y4);
  \draw[->] (x10) to (y4);
  \draw[->] (x11) to (y4);

  \draw[->, dotted] (x12) to (y5);
  \draw[->] (x13) to (y5);
  \draw[->] (x14) to (y5);

  \draw[->] (x15) to (y1);
  \draw[->] (x1) to (y1);
  \draw[->] (x2) to (y1);

  \draw[->] (x3) to (y2);
  \draw[->] (x4) to (y2);
  \draw[->] (x5) to (y2);

\end{tikzpicture}} &
{
\begin{tabular}{c}
	\begin{tabular}{l|c|c}
	& cyan & yellow\\
	\hline
	cyan & $\langle2, \dots, 2, 3\rangle$ & $\langle 1, 2, 3, 3, 3\rangle$ \\
	\hline
	yellow  &$\emptyset$ & $\langle 1, 1, 1, 1 ,1\rangle$ \\
	\end{tabular}\\[3em]
	\raisebox{-.5\height}{\begin{tikzpicture}[scale=0.8, every node/.style={scale=0.8}]
  \node[draw, circle, fill=cyan] (x) at (0,3) {$x$};
  \node[draw, circle, fill=yellow] (y) at (0,0) {$y$};

  \draw[->, loop left, out=-135, in=135, looseness=10] (x) to (x);
  \draw[->, loop right, out=-45, in=45, looseness=10] (x) to (x);

  \draw[->] (x) to[bend left] (y);
  \draw[->] (x) to (y);
  \draw[->] (x) to[bend right] (y);

  \draw[->, loop below, looseness=10, out=-45, in=-135] (y) to (y);
\end{tikzpicture}}
\end{tabular}
}
\end{tabular}
\caption{\label{fig:qf-build-example}An application of Algorithm~\ref{algo:build-qf}.}
\end{figure*}

If we try blindly to apply Algorithm~\ref{algo:qf-to-fib} to make it into a fibration, we do not obtain exactly what we may expect 
(that is, the deletion of the blue arc and the addition of the missing dotted arcs).
What happens will depend on the non-deterministic choices that the algorithm performs: Figure~\ref{fig:qf-build-example-fixed}
shows one of the possible results obtained. 
We intentionally started with a graph whose symmetries were also explainable by automorphisms: the algorithms do not reconstruct
the original graph (and in fact the final result is rigid), but it fully reconstructs its fibration symmetries. 

\begin{figure}
\centering
\begin{tikzpicture}[scale=0.8, every node/.style={scale=0.8}]
\def \n {15}
\def \radius {4cm}
\def \margin {6} % margin in angles, depends on the radius
\foreach \s in {1,...,\n}
{
  \node[draw, circle, fill=cyan] (x\s) at ({360/\n * (\s - 1)}:\radius) {$x_{\s}$};
  \draw[->] ({360/\n * (\s - 1)+\margin}:\radius) 
    arc ({360/\n * (\s - 1)+\margin}:{360/\n * (\s)-\margin}:\radius); 
}
  \draw[->] (x1) to[bend right=70] (x3);
  \draw[->] (x2) to[bend right=70] (x4);
  \draw[->] (x3) to[bend right=70] (x5);
  \draw[->] (x4) to[bend right=70] (x6);
  \draw[->] (x5) to[bend right=70] (x7);
  \draw[->] (x6) to[bend right=70] (x8);
  \draw[->] (x7) to[bend right=70] (x9);
  \draw[->] (x8) to[bend right=70] (x10);
  \draw[->] (x9) to[bend right=70] (x11);
  \draw[->] (x10) to[bend right=70] (x12);
  \draw[->] (x11) to[bend right=70] (x13);
  \draw[->] (x12) to[bend right=70] (x14);
  \draw[->] (x13) to[bend right=70] (x15);
  \draw[->] (x14) to[bend right=70] (x1);
  \draw[->] (x15) to[bend right=70] (x2);
\def \n {5}
\def \radius {2cm}
\def \margin {14} % margin in angles, depends on the radius
\foreach \s in {1,...,\n}
{
  \node[draw, circle, fill=yellow] (y\s) at ({360/\n * (\s - 1)}:\radius) {$y_{\s}$};
  \draw[->] ({360/\n * (\s - 1)+\margin}:\radius) 
    arc ({360/\n * (\s - 1)+\margin}:{360/\n * (\s)-\margin}:\radius);
}

  %\draw[->, blue] (x8) to[bend right=30] (x1);

  \draw[->, red] (x1) to (y3);
  \draw[->, red] (x2) to (y3);

  \draw[->] (x8) to (y3);

  \draw[->] (x9) to (y4);
  \draw[->] (x10) to (y4);
  \draw[->] (x11) to (y4);

  \draw[->, red] (x4) to (y5);
  \draw[->] (x13) to (y5);
  \draw[->] (x14) to (y5);

  \draw[->] (x15) to (y1);
  \draw[->] (x1) to (y1);
  \draw[->] (x2) to (y1);

  \draw[->] (x3) to (y2);
  \draw[->] (x4) to (y2);
  \draw[->] (x5) to (y2);

\end{tikzpicture}
\caption{\label{fig:qf-build-example-fixed}An application of Algorithm~\ref{algo:qf-to-fib} to the quasifibration obtained after 
applying Algorithm~\ref{algo:build-qf} to Figure~\ref{fig:qf-build-example}. The newly added arcs are in red, whereas the old blue arc is deleted.}
\end{figure}
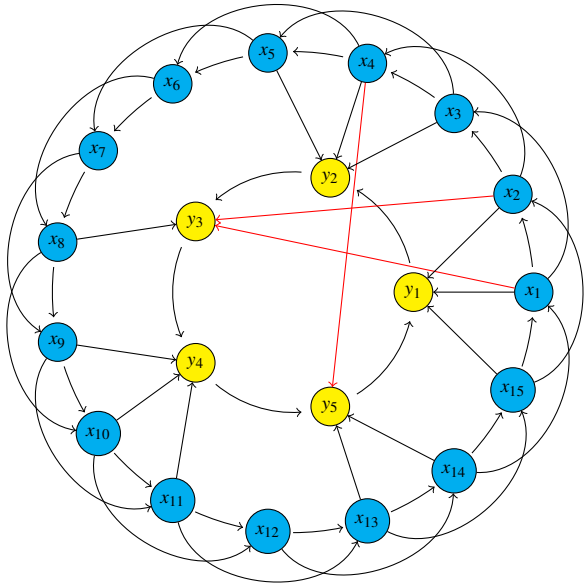

\begin{algorithm}[H]
    \SetAlgoLined
    \LinesNumbered
    \SetKwInOut{Input}{Input}
    \SetKwInOut{Parameter}{Parameter}
    \SetKwInOut{Output}{Output}
    \Input{a graph $G$}
    \Parameter{three positive integers $t, m_1, m_2$ (with $m_1\leq m_2$)}
    \Output{an equivalence relation $\sim$ on $N_G$}
    
    \For{every node $x\in N_B$} {
        $T_x \leftarrow \view Gx\upharpoonright t$\;
    }
    
    \For{every pair of nodes $x,y\in N_G$} {
        $M_{xy} \leftarrow \text{tree edit-distance between $T_x$ and $T_y$}$\;
        \tcc{Try unordered; then ordered, if time-out expires}
    }
    
    normalize $M$ so that its $L_1$-norm is 1\;
    let $z_t$ be the number of non-isomorphic views of height $t$\;
    
    \For{$c=m_1,\dots,\min(m_2,z_t)$} {
        perform agglomerative clustering of $N_G$ with $c$ clusters using distance matrix $M$\;
        let $\sim_c$ be the resulting equivalence relation\;
        let $s_c$ be its silhouette index\;
    }
    
    \Return{$\sim_c$ with maximum $s_c$ value}
    \caption{\label{algo:equiv-rel}{\scshape BuildEquiv}: This algorithm builds an equivalence relation on the nodes of $G$, trying to reconstruct its fibration symmteries.}
\end{algorithm}

\subsection{Step (1): The {\scshape BuildEquiv} algorithm}
\label{sec:equiv-rel}

Step (1) of our plan requires to obtain an equivalence relation on the
nodes of $G$ that is as close as possible to the original symmetries
in the graph without noise (that is, $\mineq_H$).  As we said, this
task is mostly heuristic, and based on the idea that a direct
application of the minimum fibration construction $\mineq_G$ (see
Section~\ref{sec:ufc}) would not work because of the presence of
noise.  We took into consideration many alternatives, and finally
chose one that is at the same time consistent with our view of the
phenomenon under description and that turned out to be almost always
the best alternative in our experiments. We postpone to
Section~\ref{sec:experiments} a description of some alternative
heuristics.

Our final solution (Algorithm~\ref{algo:equiv-rel}) is to relax the
construction of $\mineq_G$ (that puts together two nodes if and only
if their views are isomorphic) by considering instead \emph{tree
  edit-distance}~\cite{bille2005survey}, a generalization to trees of
the standard Levenshtein distance between strings\footnote{Although
  general edit distance is defined for trees with labels on the tree
  nodes, our trees are unlabelled (equivalently, all tree nodes have
  the same label), so the only operations allowed are deletion or
  insertion of a subtree.}. Unfortunately, the general problem of
computing tree edit-distance on unordered trees (i.e., trees where the
order of chidren is irrelevant, like ours) is unlikely to be
solvable in polynomial time: it is MAX SNP hard, hence not
even approximable with a PTAS unless P=NP. On small trees, though,
the computation is possible with dynamic programming as described
in~\cite{yoshino2013dynamic}.  On the other hand, tree edit-distance
for ordered trees can be solved efficiently, and efficient
well-documented implementations exist~\cite{paassen2015toolbox} (see
also~\cite{zhang1989simple}).

Benjamin Paassen, the author of the above toolbox, provided us an experimental implementation of the algorithm of~\cite{yoshino2013dynamic},
and also of Zhang's polynomial-time constrained version that computes an upper bound of the unordered edit distance. 

After some experimentation with combining the various algorithms at our disposal, we observed that the best 
solution for our problem was to run the dynamic-programming algorithm of~\cite{yoshino2013dynamic} with a time-out,
and to resort to ordered (exact) edit distance if the time-out expires. 

The (mixed) tree edit-distances computed this way provide a distance matrix that is fed to a clustering algorithm. We employed hierarchical agglomerative clustering~\cite{nielsen2016introduction}
with average distance between clusters as linkage option\footnote{We used the \texttt{scikit-learn} implementation (AgglomerativeClustering).}.
Since it is impossible to establish a reasonable threshold for distances (they are too dependent on the actual graph structure and density),
we fix the number of clusters $c$ instead, and try for different values of $c$.
We compute the silhouette coefficient~\cite{rousseeuw1987silhouettes} for each $c$, to see how much the clustering obtained fits the original
distances, and select the clustering with largest silhouette: this is in fact one of the many possible techniques
generally adopted in data analysis to select the number of clusters\footnote{The problem is very well discussed and no general solution
exists~\cite{wierzchon2018modern}: in most cases a manual inspection is required (for instance, in the well known ``elbow'' method).
The use of indices like the silhouette coefficient to try to automatize this analysis is certainly error-prone, and may
produce suboptimal results.}. 

One useful note on clustering is the following: if there are $z_t$ non-i\-so\-mor\-phic views of depth $t$, having more than $z_t$ clusters can only decrease
the silhouette coefficient (because the clustering will be obliged to break classes containing trees that are at distance zero). 
The value $z_t$ can be computed by running the first $t$ steps of Cardon-Crochemore and looking at how many classes are found.

Observe that if $x \mineq_G y$ then $\view Gx \cong \view Gy$, hence their tree edit-distance (at any depth) is zero;
so the two nodes should end up in the same cluster (provided that we are using at least as many clusters as the
number of classes in $\mineq_G$). For this reason, in non-degenerate cases the output of Algorithm~\ref{algo:equiv-rel}
will always be a coarsening of $\mineq_G$.

\section{Experiments and Discussion}
\label{sec:experiments}

The purpose of this section is to justify the heuristic choices of
Algorithm~\ref{algo:equiv-rel} based on some experiments performed on
real as well as synthetic datasets. In all cases, we will have a
ground truth that is obtained by human inspection in the former case
whereas it comes from the way synthetic datasets are built in the
latter.

\paragraph{Performance metrics.}
Since what we want to solve is basically a clustering problem, one of the
main ingredients that we will use in our experimental analysis is a
clustering-comparison metrics w.r.t.~the ground truth. We decided to
adopt uniformly the AMI score, a version of the standard NMI
(Normalized Mutual Information) score adjusted for
chance~\cite{vinh2010information}, as implemented in the
\texttt{sklearn.metrics} package.  We will write $\ami(\sim_1,\sim_2)$
to denote the AMI score between $\sim_1$ and $\sim_2$: recall that
this value is $1$ if $\sim_1$ and $\sim_2$ coincide, whereas random
independent clustering have an expected AMI around 0 on average (the score itself can be negative).
 
We tried other alternatives (like the adjusted Rand index) that offer
more or less consistent results, although on our datasets they tend to
penalize more even small classification errors.
   
\paragraph{Real-world datasets.}
We used a real-world dataset to test the algorithms. We use the
connectome of the worm {\it C. elegans} which is a fully mapped neural
system of a model organism. The connectome consist of 302 neurons, and
here we consider the set of neurons with their chemical synapses
connections that are activated when the worm is moving forward and
backward, separately. It has been shown in~\cite{morone2019NatComm}
that these networks are characterized by pseudo-symmetries, i.e.,
almost automorphisms.  Below, we use these two networks to test the
algorithms and show that, more generally, they are characterized by
quasifibration. In~\cite{morone2019NatComm} the authors have presented a
manually-repaired ideal network with perfect symmetries and here we use
this ideal network as a 'ground truth' to compare with the results of
the present algorithms.

\paragraph{Synthetic datasets.}
In order to build synthetic datasets, we should produce a fibration-rich graph $H$ and introduce some noise in it.
We proceed in the following way:
\begin{itemize}
  \item We started from a directed scale-free graph with $n$ nodes~\cite{bollobas2003directed}, and produced its 
  minimum base, that is fibration prime, and we built a graph $H$ that is fibred over it: for every node of the base we decided at random the size of its fibre
  (an integer between $v_{\mathrm min}$ and $v_{\mathrm max}$) and established how to lift nodes arbitrarily 
  but so that the resulting assignment is a fibration.  
  \item We added and/or removed $s$ arcs from $H$, obtaining the graph $G$ whose nodes will be clustered. We use $\mineq_H$ as ground truth.
\end{itemize}
The above process has a number of parameters that will influence the result, besides the obvious fact that the whole
construction is probabilistic. 

\paragraph{Methods.}
For our comparison, we consider the following alternatives:
\begin{itemize}
  \item \emph{Cardon-Crochemore:} we used $\sim_G$ as a baseline, to see if and how much we improve over it; 
  \item \emph{Variants:} we took into consideration some possible variants of Algorithm~\ref{algo:equiv-rel}; in particular:
  	\begin{itemize}
  	  \item \emph{Linkage type:} one crucial aspect in agglomerative clustering is deciding how the distance between two
  	  clusters is computed; we considered three of the most common alternatives in the literature~\cite{wierzchon2018modern}:
  	  ``single'', ``complete'' and ``average'' (the distance between two clusters is defined as the minimum, maximum, average distance between
  	  two points in the two clusters); 
  	  for the sake of simplicity, we do not show the results of this analysis here, but 
  	  report that ``single'' always outperforms (or is equivalent to) the other linkage types;
  	  since it is also more efficient, it is our final choice in Algorithm~\ref{algo:equiv-rel};
  	  \item \emph{Number of clusters:} as explained in Section~\ref{sec:equiv-rel}, selecting the number of clusters is 
  	  a hard and delicate task; since the silhouette method is normally not adopted without human intervention, we decided
  	  to compare the results obtained by Algorithm~\ref{algo:equiv-rel} with those that would be obtained if we knew
  	  the real number of clusters (the number of clusters in the ground truth). 
	\end{itemize}
  \item \emph{Centrality-based clustering:} a naive, simple alternative to Algorithm~\ref{algo:equiv-rel} would be to use some diffusion-based centrality algorithm \emph{\`a la} Katz~\cite{KatNSIDSA}
  and then to cluster the nodes based on their centrality score, e.g. using K-means; note that again, since the number of clusters is unknown,
  we decide to use the number of clusters from the ground-truth; we ran in fact a number of centrality algorithms~\cite{BoVAC} and, for the
  ones that have a parameter, we used the value that produced the maximum $\ami$ with the ground truth. 	  
\end{itemize}

\subsection{Full Reconstruction}

\begin{figure*}
\centering
\includegraphics[scale=.6]{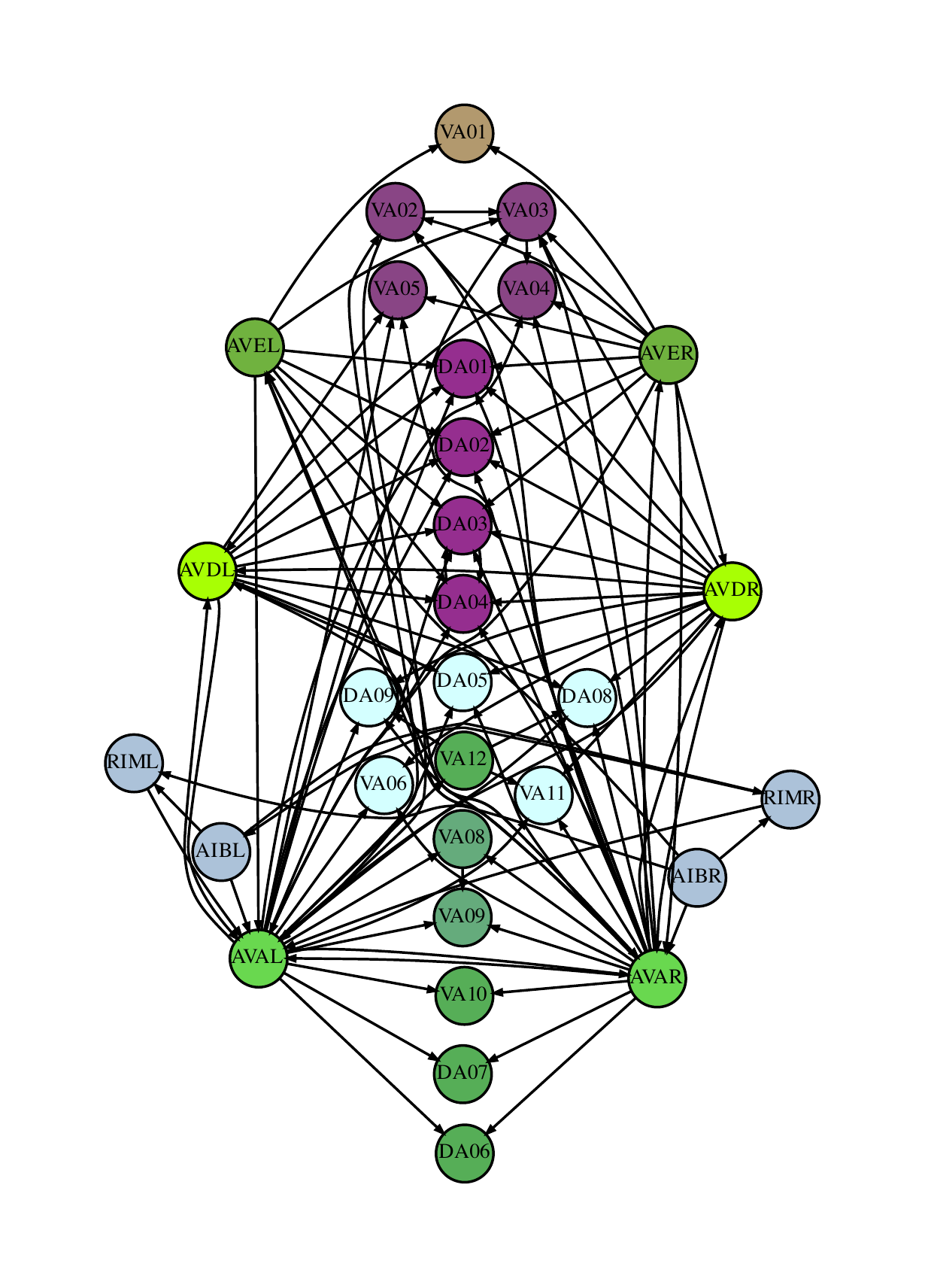}
\includegraphics[scale=.6]{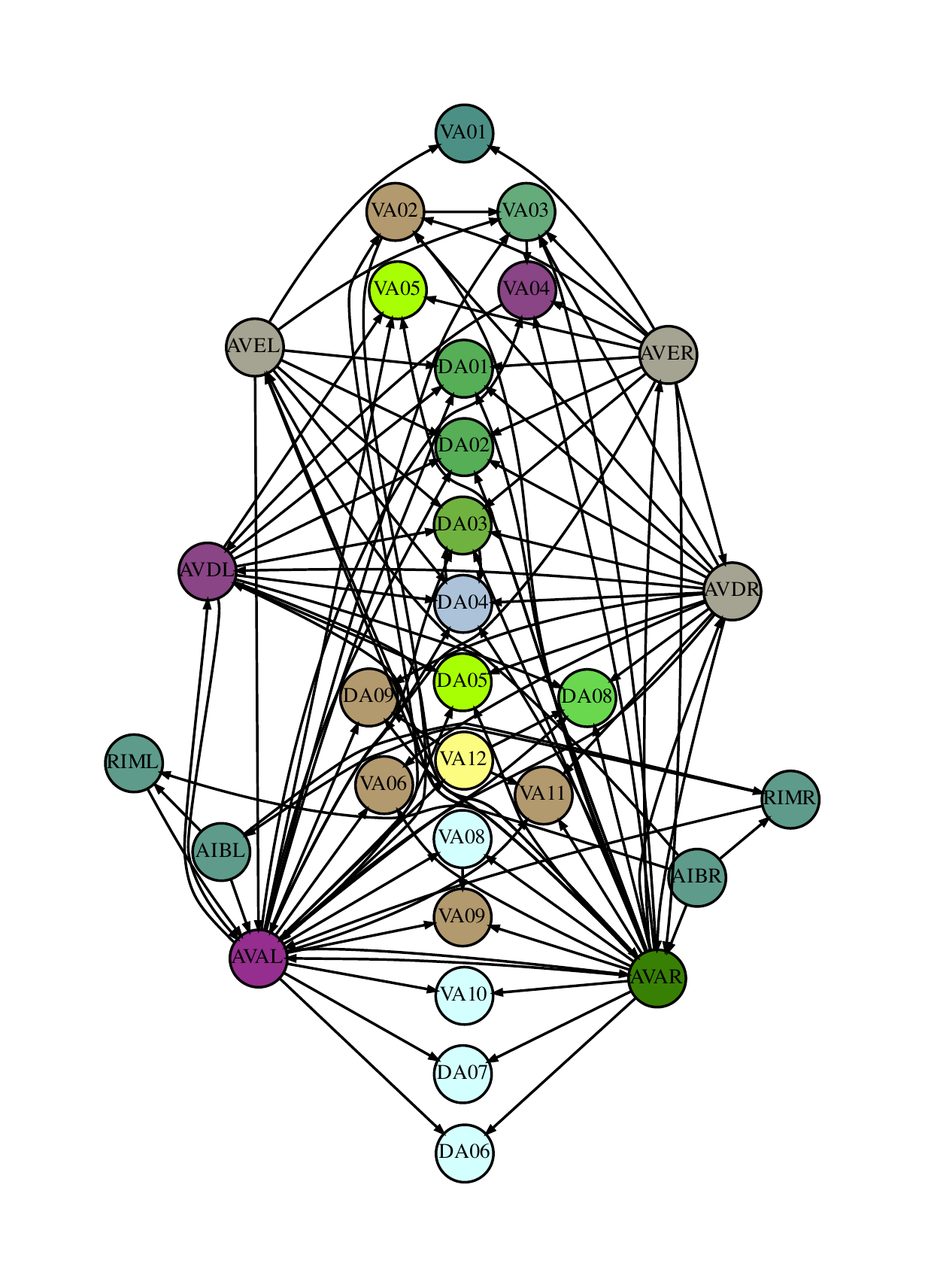}
\caption{\label{fig:back3}The backward network, colored with the human
  crafted ground truth from \cite{morone2019NatComm} (left) and with
  the equivalence relation obtained by Algorithm~\ref{algo:equiv-rel}
  (right) with depth 3 (adjusted mutual information is $0.417$).}
\end{figure*} 
\begin{figure*}
\centering
\includegraphics[scale=.25]{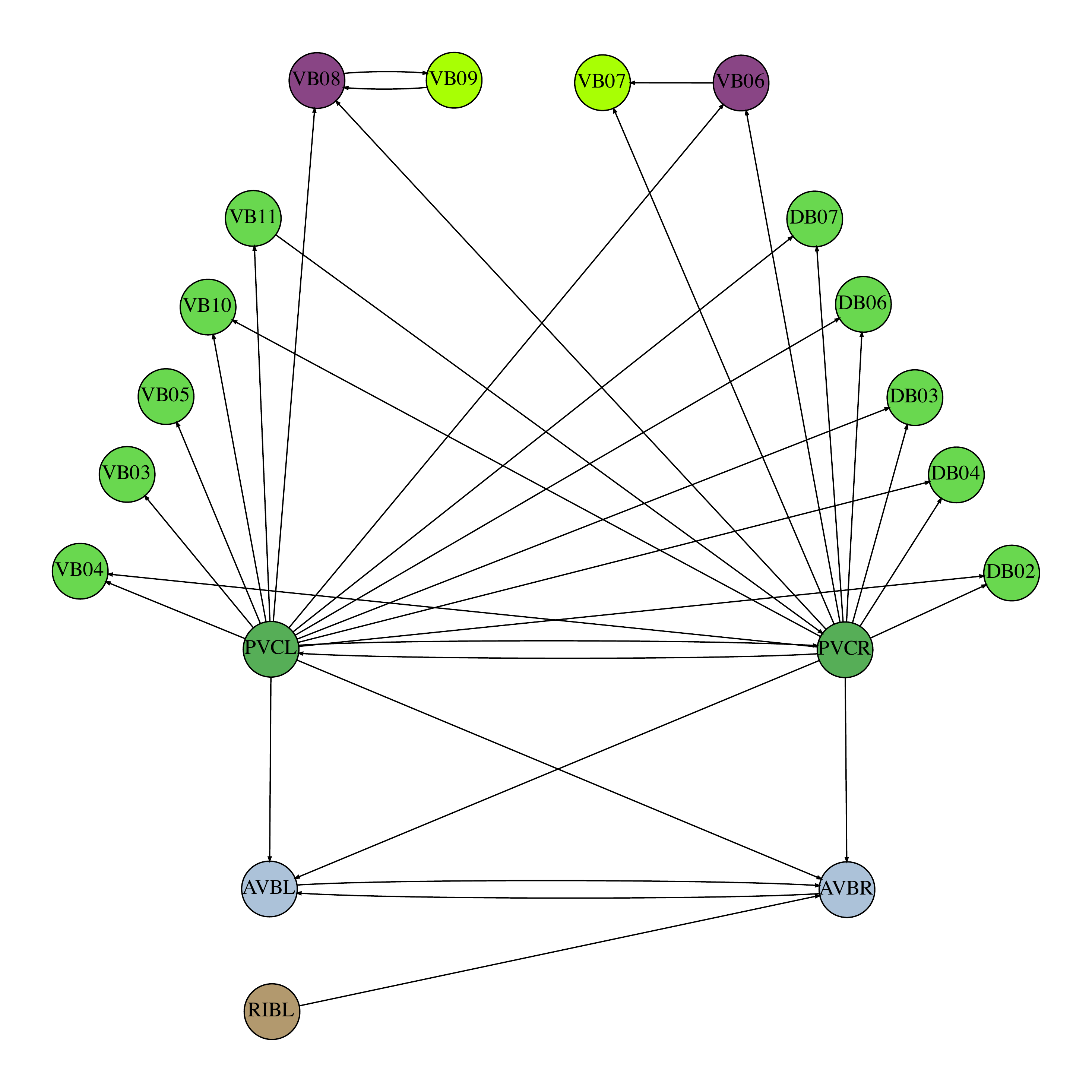}
\includegraphics[scale=.25]{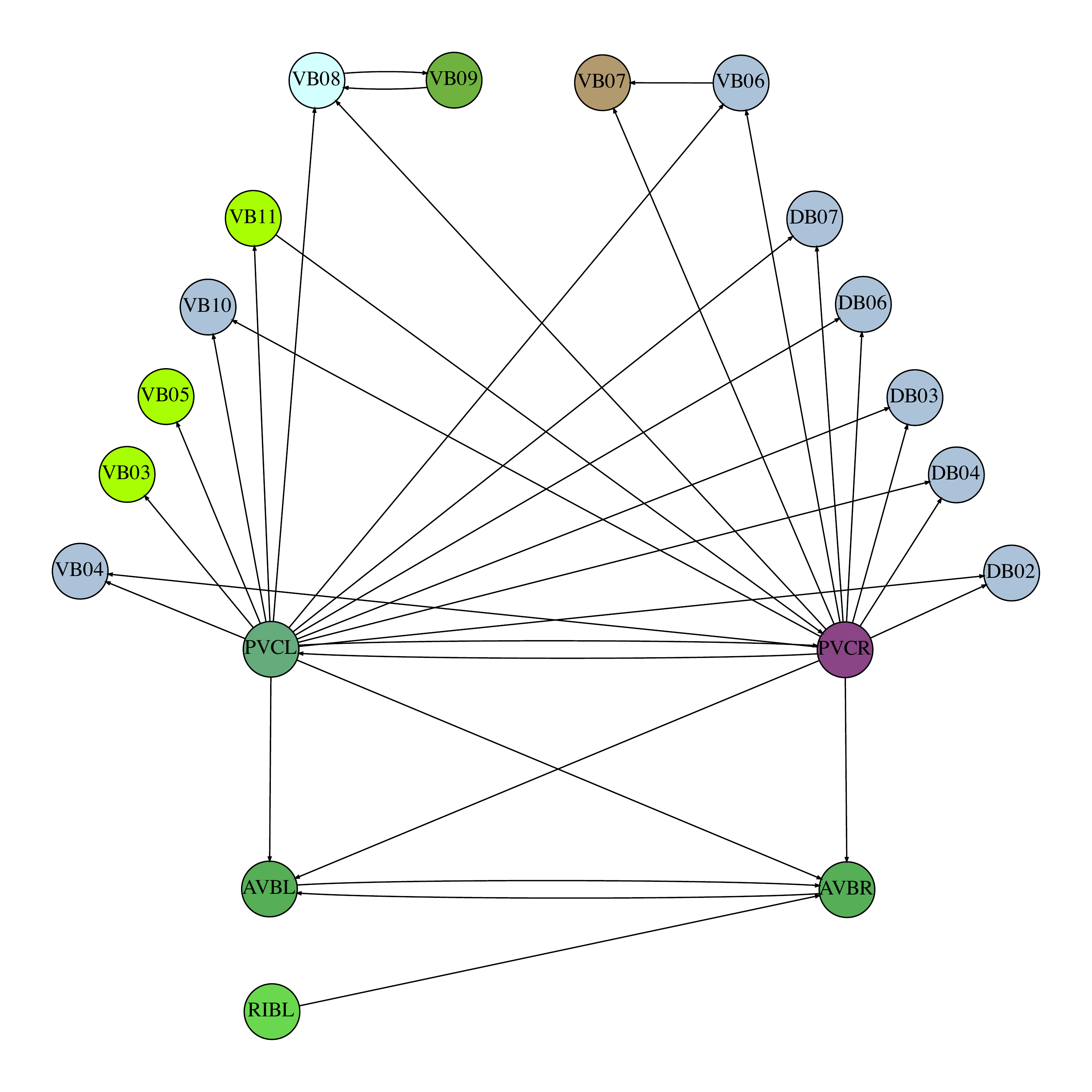}
\caption{\label{fig:forw4}The forward network, colored with the human
  crafted ground truth from \cite{morone2019NatComm} (left) and with
  the equivalence relation obtained by Algorithm~\ref{algo:equiv-rel}
  (right) with depth 4 (adjusted mutual information is $0.594$).}
\end{figure*} 

In this section we shall consider two examples of full reconstruction:
we will start from two real-world graph $G$ and apply the whole
process described in Section~\ref{sec:tools} to both of them.  In the
case we consider here, the graphs represent two neural circuits controlling the
locomotion function of {\it C. elegans} as defined in
Ref. \cite{morone2019NatComm}: the forward and backward circuits. The
backward graph has 30 nodes and 106 arcs, with an average degree of
$3.53$; the forward graph has 19 nodes and 35 arcs, with an average
degree of $1.84$.

For these datasets, we have a handcrafted ground-truth based on the
supposed function of each single node obtained in
\cite{morone2019NatComm}; the ground-truth for the backward graph
outlines 10 classes with various sizes (the largest class contains 5
nodes, whereas five classes contain 1 or 2 nodes); the ground-truth
for the forward graph has 6 classes (the largest class contain 10
nodes, whereas five cases contain 1 ore 2 nodes).  The ground-truth
is not provided to the algorithm, though, which just starts with the
original graph itself.

We tried our experiment with a depth (parameter of
Algorithm~\ref{algo:equiv-rel}) of 2, 3, 4 and 5. While there is a
clear improvement in the performances as the depth increases in the
case of the Forward network, the converse seems to happen for the
Backward network.  An inspection of the logging shows, however, that
the number of times the computation of the unordered tree-distance is
timed out grows: $60\%$ of the entries are timed out when the depth is
$4$ and more than $80\%$ when the depth is $5$. The frequency of this
event grows as the depth gets larger, because trees grow in size.

\begin{tabular}{l||c|c||c|c}
& \multicolumn{2}{c||}{Backward} & \multicolumn{2}{c}{Forward}\\
& \# clusters & $\ami$ & \# clusters & $\ami$\\
\hline
Depth 2 & 18 & 0.262 & 3 & 0.423\\
Depth 3 & 15 & 0.417 & 5 & 0.500\\
Depth 4 & 14 & 0.398 & 9 & 0.594\\
Depth 5 & 6 & 0.355 & 9 & 0.594\\
\hline
Cardon-Crochemore & 10 & 0.295 & 10 & 0.543
\end{tabular}  

\smallskip
The difference between the ground truth and the equivalence relation obtained at the end is shown in Figure~\ref{fig:back3} and Figure~\ref{fig:forw4}.

Running Algorithm~\ref{algo:build-qf}, we discover that the total
error of the reconstructed quasifibration $\xi: G \to B$ is 18 (three
arcs are in excess and fifteen are missing). The final fibration
$\xi': G' \to B$ obtained by Algorithm~\ref{algo:qf-to-fib} modifies
the graph $G$ as shown in Figure~\ref{fig:back3-diff}.

\begin{figure}
\centering
\includegraphics[width=.5\textwidth]{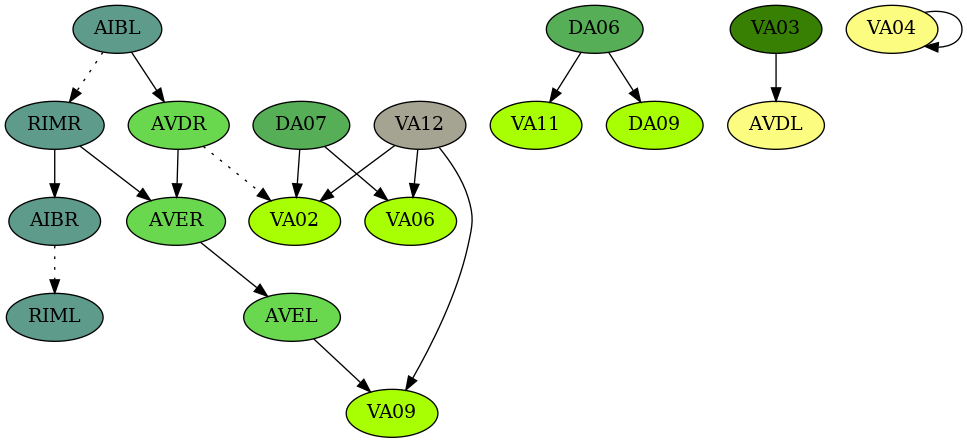}
\caption{\label{fig:back3-diff}The deleted (dashed) and added arcs that should be used to modify the graph $G$ of Figure~\ref{fig:back3} (right) to
turn the coloring into a local in-isomorphism.} 
\end{figure} 

\subsection{Synthetic datasets}

We ran synthetic experiments as described above, with $n=30$, $v_{\mathrm min}=1$, $v_{\mathrm max}=5$ and $s=5$,
and with depth 2 and 3.
Here are the main characteristics of the generated graphs:

\begin{tabular}{l|r|r|r}
& \# nodes & \# arcs & \# classes in the ground truth\\
\hline
Average & $25.5$ & $173.7$ & $8.5$\\
Std     & $6.5$  & $53.9$  & $1.7$\\
Min     & $9$    & $53$     & $4$\\
Max     & $49$    & $373$    & $15$\\
\end{tabular}  

\medskip
As the reader can see, the size, density and number of classes in the
ground truth match those of the real-world datasets.
Figure~\ref{fig:synth-30-5-depth3-boxplot} shows the boxplot with the
performances of the various clustering techniques mentioned above in the case of depth 3.
``Reduced UED/OED'' is the final result of the algorithm in the paper
(after step (4) of Section~\ref{sec:tools} is applied) when
Unordered/Ordered Edit Distance is used and the silhouette method is
applied to obtain the number of clusters.  ``Agglomerative UED/OED
exact'' uses the technique of Algorithm~\ref{algo:equiv-rel} but
avoids the use of silhouette and directly applies the number of
clusters from the ground truth.

\begin{figure}
\centering
\includegraphics[width=.5\textwidth]{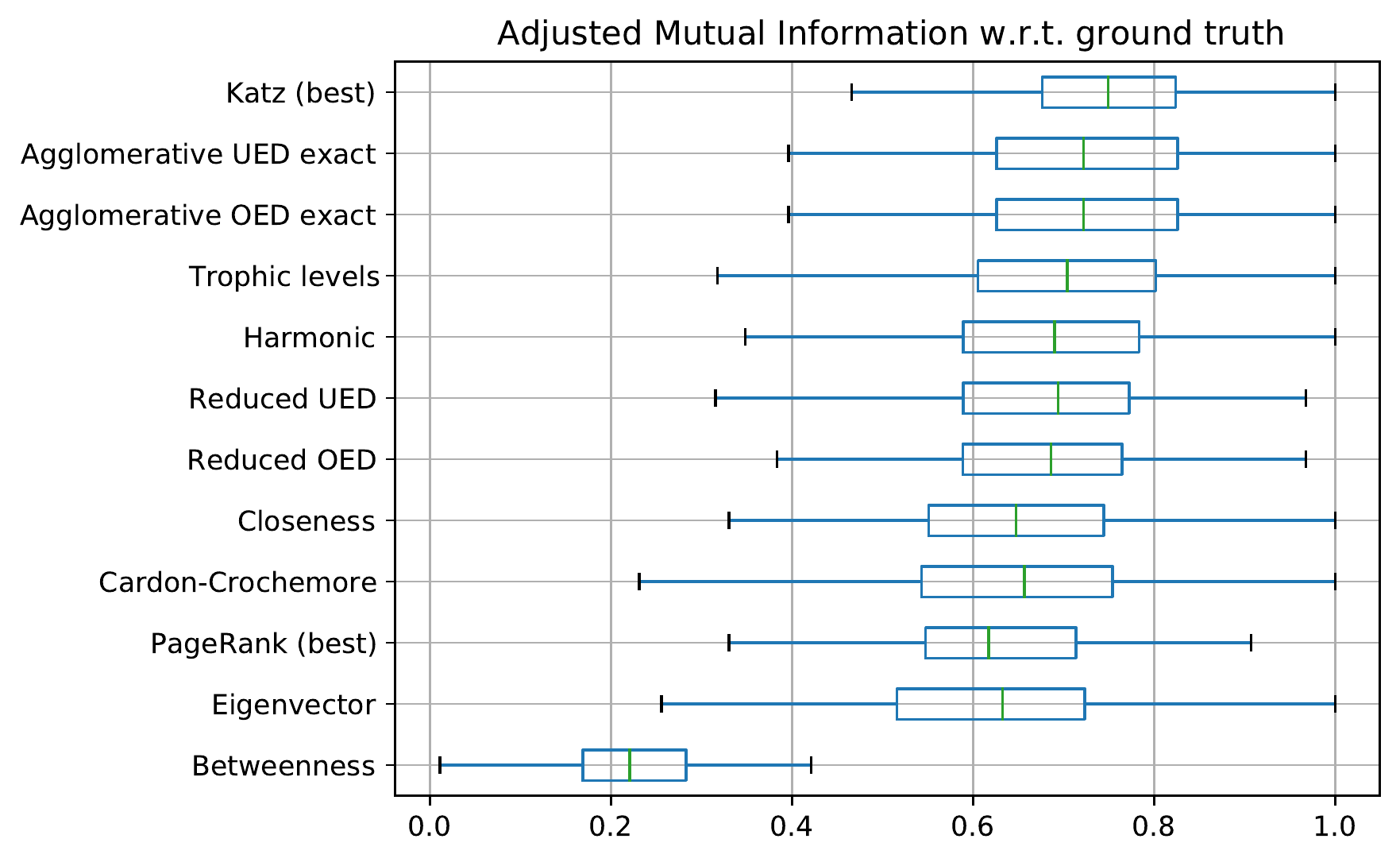}
\caption{\label{fig:synth-30-5-depth3-boxplot}Boxplot with the Adjusted Mutual Information for the clustering generated by various methods;
note that the only ones that do not use the number of classes in the ground truth are those called ``Reduced UED/OED'' 
and ``Cardon-Crochemore''.}
\end{figure}

It is worth observing that ``Reduced UED'' is always better than ``Cardon-Crochemore''. A more detailed analysis shows that the ratio between
``Reduced'' and the ground truth outperforms ``Cardon-Crochemore'' of $19\%$.
The discussed behaviour can also be seen in Figure~\ref{fig:synth-30-5-heat}, where we produce two heatmaps (one for depth 2 and one for depth 3);
each point in each heatmap represents the result of an experiment performed: its  X and Y coordinates are the Adjusted Mutual Information of ``Cardon-Crochemore'' and ``Reduced UED'', respectively.
We verify that the mass is on or above the diagonal, and that this situation improves as the depth is increased.  

\begin{figure}
\centering
\includegraphics[scale=.5]{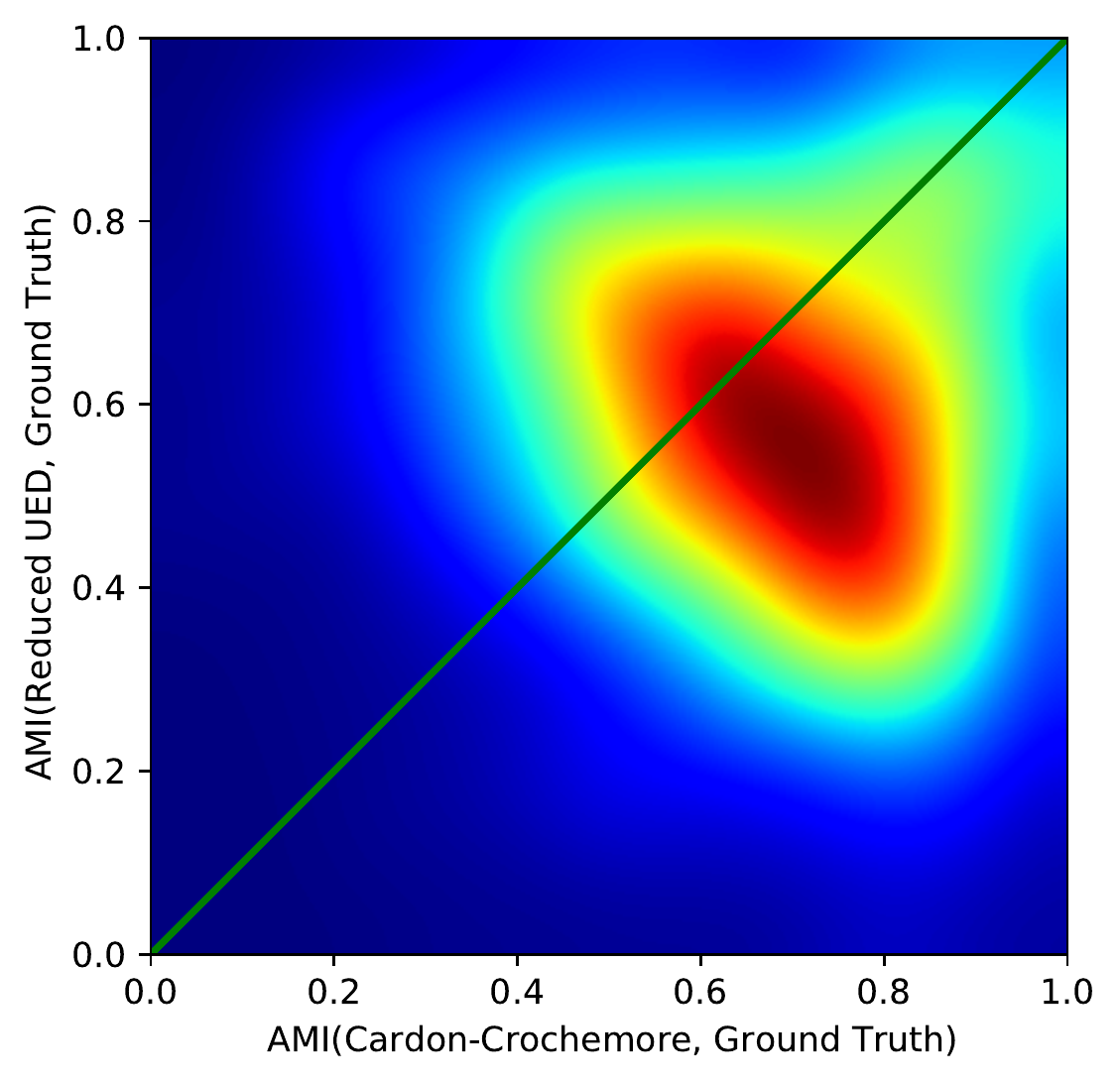}
\includegraphics[scale=.5]{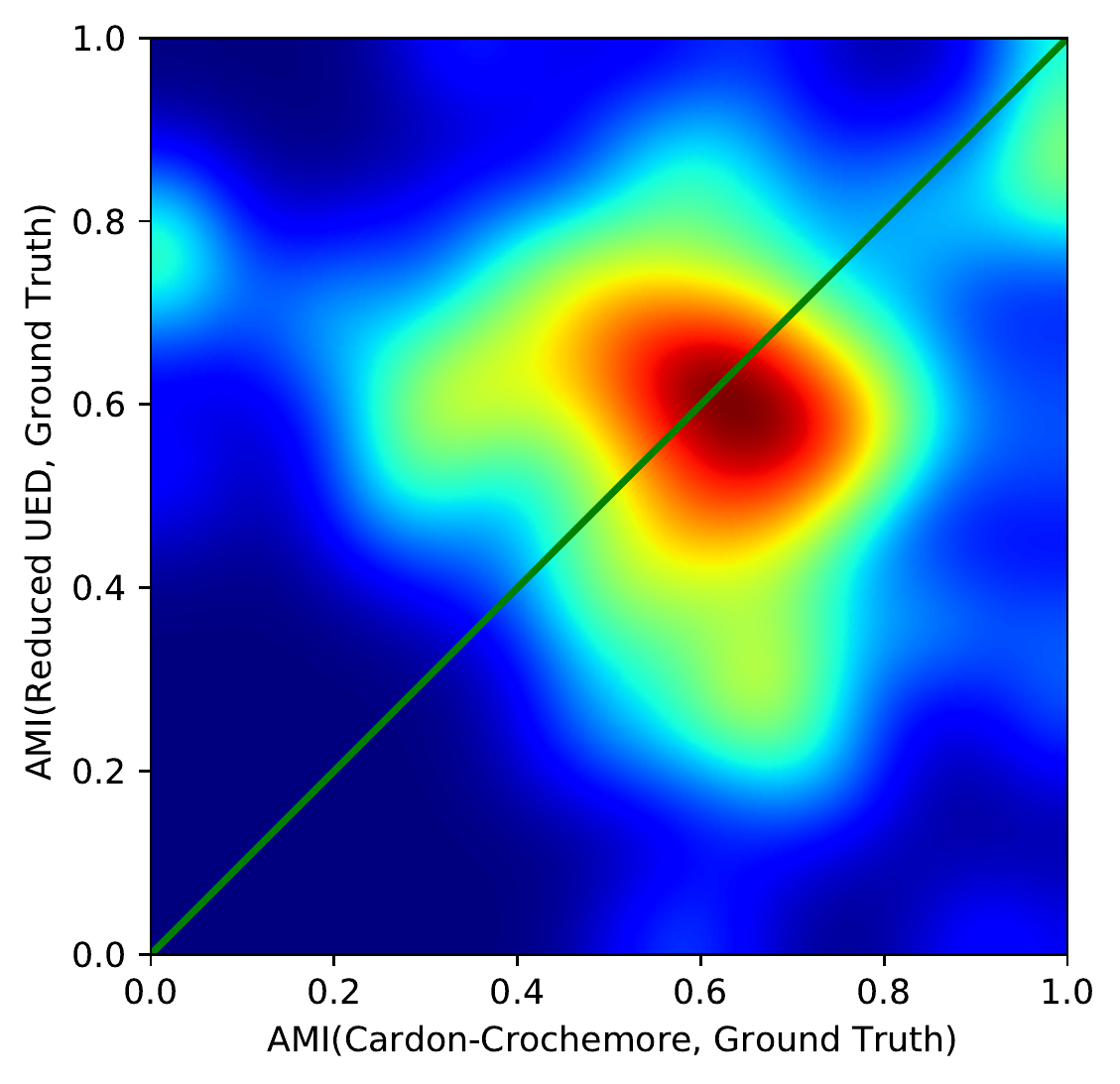}
\caption{\label{fig:synth-30-5-heat}Heatmap showing a comparison between Adjusted Mutual Information for ``Cardon-Crochemore'' (X axis) and ``Reduced UED'' (Y axis),
for depth 2 (top) and 3 (bottom). The mass is concentrated over the diagonal, and the improvement of ``Reduced UED'' over ``Cardon-Crochemore'' increases
with the depth (the core moves up).}
\end{figure}

\section{Conclusion}

Imperfections in real biological data is a major hurdle in biological network analysis. In this paper we present an approach to dealing with these imperfections via quasifibrations of graphs. We developed a four-step algorithm restoring the graph to a more symmetric version. First, we find the equivalence relations combining "almost" symmetric nodes. Second, we build a quasifibration using the equivalence classes identified in the first step. Third, we construct a graph by modifying a (provably) minimal amount of links in the original graph for which this quasifibration is a fibration. Last, we compute the minimal fibration on the obtained graph to coarsen the obtained result. It was shown analytically that steps 2-4 are performed in an optimal manner. First step is done employing agglomerative clustering.

We used synthetic datasets and real-world datasets studied in \cite{morone2019NatComm} to fine-tune the clustering method. Performance of the different variations was assessed using AMI (adjusted mutual information). Unordered tree edit-distance accompanied by single linkage with the number of clusters corresponding to the highest value of the Silhouette coefficient have shown the best performance of all considered cases. The algorithm outperformed Cardon-Crochemore by 19\% on average on the synthetic networks. On both real networks the algorithm showed a result better than Cardon-Crochemore by comparing outputs with manually curated results in \cite{morone2019NatComm}.

\section*{Acknowledgements}

We want to thank Sebastiano Vigna for many helpful discussions and Benjamin Paassen for providing us the implementations
of unordered tree-edit distance and for many useful insights on the
topic. Funding was provided by NIBIB and NIMH through the NIH BRAIN
Initiative Grant \# R01 EB028157.

\section*{Data Availability Statement}

All the algorithms are available at \url{https://github.com/boldip/qf} and \url{https://github.com/makselab/QuasiFibrations}. All data is available at \url{https://osf.io/amswe}.

%\nocite{*}
\bibliography{biblio,makse,extra}

\end{document}